\documentclass[twoside,11pt]{article}

\usepackage[abbrvbib]{jmlr2e}

\usepackage{calligra}
\usepackage{cancel}
\usepackage{graphicx,psfrag,epsf,subcaption,bbm}
\usepackage{mathtools}
\usepackage{enumerate}
\usepackage{url} 
\usepackage{microtype}

\usepackage{calligra}
\usepackage[linesnumbered,ruled,vlined]{algorithm2e}
\usepackage[noend]{algpseudocode}
\usepackage[nothm]{YK}
\usepackage[shortlabels]{enumitem}

\newcommand{\BALD}{\begin{aligned}}
	\newcommand{\EALD}{\end{aligned}}
\newcommand{\BALDS}{\begin{aligned*}}
	\newcommand{\EALDS}{\end{aligned*}}
\newcommand{\BCAS}{\begin{cases}}
	\newcommand{\ECAS}{\end{cases}}
\newcommand{\BEAS}{\begin{eqnarray*}}
	\newcommand{\EEAS}{\end{eqnarray*}}
\newcommand{\BEQ}{\begin{equation}}
\newcommand{\EEQ}{\end{equation}}
\newcommand{\BIT}{\begin{itemize}}
	\newcommand{\EIT}{\end{itemize}}
\newcommand{\BMAT}{\begin{bmatrix}}
	\newcommand{\EMAT}{\end{bmatrix}}
\newcommand{\BNUM}{\begin{enumerate}}
	\newcommand{\ENUM}{\end{enumerate}}

\def\reals{\mathbf{R}}

\newcommand{\deltavec}{\mbox{\boldmath $\delta$}}

\newcommand{\ta}{\tilde{a}}

\newcommand{\thetas}{\theta^*}
\newcommand{\deltas}{\delta^*}

\newcommand{\tdelta}{\widetilde{\delta}}

\newcommand{\ttheta}{\widetilde{\theta}}

\newcommand{\tthetadk}{\widetilde{\theta}_k^d}
\newcommand{\Thetah}{\widehat\Theta}

\newcommand{\htau}{\hat{\tau}}
\newcommand{\drho}{\dot{\rho}}
\newcommand{\ddrho}{\ddot{\rho}}

\newtheorem{result}[theorem]{Result}
\newcommand{\indicator}{\mathbf{1}}
\usepackage{wrapfig}
\usepackage{enumitem}
  \setlist{leftmargin=*,noitemsep}

\newenvironment{sproof}[1]{\noindent {\textit{#1.}}}{\hfill \BlackBox\\[2mm]} 
\def\hdelta{\widehat{\delta}}
\def\htheta{\widehat{\theta}}

\usepackage{lastpage}
\jmlrheading{23}{2022}{1-\pageref{LastPage}}{7/21; Revised
2/22}{4/22}{21-0739}{Subha Maity, Yuekai Sun, and Moulinath Banerjee}
\ShortHeadings{Meta-analysis of heterogeneous data}{Maity, Sun and Banerjee}

\begin{document}

\title{Meta-analysis of heterogeneous data: integrative sparse regression in high-dimensions}
\author{\name Subha Maity \email smaity@umich.edu \\
        \name Yuekai Sun \email yuekai@umich.edu\\
        \name Moulinath Banerjee \email moulib@umich.edu\\
       \addr Department of Statistics\\
       University of Michigan \\
       Ann Arbor, MI}
    \editor{Qiang Liu}
  \maketitle

\begin{abstract}
We consider the task of meta-analysis in high-dimensional settings in which the data sources  are similar but non-identical. To borrow strength across such heterogeneous datasets, we introduce a global parameter that emphasizes interpretability and statistical efficiency in the presence of heterogeneity. We also propose a one-shot estimator of the global parameter that preserves the anonymity of the data sources and converges at a rate that depends on the size of the combined dataset.   For high-dimensional linear model settings, we demonstrate the superiority of our  identification restrictions in adapting to a previously seen data distribution as well as predicting for a new/unseen data distribution. Finally, we demonstrate the benefits of our approach on a large-scale drug treatment dataset involving several different cancer cell-lines\footnote{Codes are available in  \href{https://github.com/smaityumich/MrLasso}{https://github.com/smaityumich/MrLasso}.}.
\end{abstract}


\section{Introduction}
\label{sec:intro}

We consider the task of synthesizing information from multiple datasets. This is usually done to enhance statistical power by increasing the effective sample size. However, as seen in many pooled studies, the inherent heterogeneity among the studies must be managed carefully. For example, in pooled genome-wide association studies (GWAS), differences in study populations and measurement methods may confound the association between the biomarkers and the outcome \citep{leek2010Tackling}. To obtain meaningful conclusions, practitioners must properly model the heterogeneity in pooled studies.

In this paper, we consider the task of fitting high-dimensional regression models to $m$ similar, but non-identical datasets. To borrow strength across the datasets, we must distinguish between the common and idiosyncratic parts of the datasets. Let $\theta_k^*\in\reals^d$ be the vector of regression coefficients for the $k$-th dataset. Without loss of generality, we posit that the $\theta_k^*$'s are related through an additive structure:
\begin{equation}
\theta_k^* = \theta_0^* + \delta_k^*,\quad k\in[m].
\label{eq:additive-decomposition}
\end{equation}
This is a decomposition of the $\theta_k^*$'s into common and idiosyncratic parts: the \emph{global parameter} $\theta_0^*\in\reals^d$ represents the common part in the datasets, while the $\delta_k^*$'s represent the idiosyncratic parts. By itself, this decomposition does not identify the common and idiosyncratic parts because the decomposition is not unique: for any $\delta\in\reals^d$, $(\theta_0^*,\delta_1^*,\dots,\delta_m^*)$ and $(\theta_0^* - \delta,\delta_1^* + \delta,\dots,\delta_m^* + \delta)$ are different decompositions of the same $\theta_k^*$'s. Thus it is necessary to impose additional identification restrictions on $\theta_0^*$ and the $\delta_k^*$'s. Note that this identification issue is not specific to the additive decomposition \eqref{eq:additive-decomposition}; any parameterization of the $\theta_k^*$'s that introduces a global parameter (to model the common parts of the datasets) is overparameterized/underdetermined. 

As we shall see, the choice of identification restriction is crucial to the efficacy of meta-analysis; an improperly defined global parameters may \emph{totally negate its benefits}. For example, consider the common ANOVA identification restriction:
\[\textstyle
\sum_{k=1}^m\delta_k^* = 0.
\]
This restriction implies that the global parameter is the average of the $\theta_k^*$'s: 
\begin{equation}\textstyle
\theta_0^* \triangleq \frac1m\sum_{k=1}^m\theta_k^*.
\label{eq:anovaIdRestriction}
\end{equation}
In some cases (\eg\ when the $\theta_k^*$'s have disjoint supports), the global parameter can be up to $m$-times denser than the local parameters, which negates any statistical efficiency gains from borrowing strength across the datasets. We defer a more comprehensive discussion of the inadequacies of common identification restrictions to Section \ref{sec:identifiability}.

In this paper, we define the global parameter as 
\begin{equation}
    \textstyle
(\theta_0^*)_j \triangleq \argmin_{\mu_j\in\reals}\sum_{k=1}^m\Psi((\theta_k^*)_j - \mu_j),\quad j\in[d],
\label{eq:redescendind}
\end{equation}
where $\Psi$ is a \emph{re-descending loss function} \citep{huber1964Robust}. Intuitively, this definition separates $(\theta_1^*)_j, \dots, (\theta_m^*)_j$ into outlier and inliers and defines $(\theta_0^*)_j$ as the average of the inliers. Thus the definition identifies the location of the inliers as the common part of the datasets. This identification restriction is motivated by an $\eps$-contamination model for the $\theta_k^*$'s:
\begin{equation}
    \textstyle
    (\theta_k^*)_j \sim (1-\eps) F_j + \eps G_j,
    \label{eq:gen-contamination}
\end{equation}
where $\eps < \frac12$ is the fraction of outliers, $F_j$ is the distribution of the inliers, and $G_j$ is the distribution of the outliers (among the $j$-th regression coefficients). The two main benefits of this identification restriction are
\begin{enumerate}
\item \emph{interpretability:} if a majority of $(\theta_1^*)_j,\dots,(\theta_m^*)_j$ are zero, then $(\theta_0^*)_j$ is zero. More generally, if there is a majority value among $(\theta_1^*)_j,\dots,(\theta_m^*)_j$ (i.e. a common value shared by more than half of $(\theta_1^*)_j,\dots,(\theta_m^*)_j$), then $(\theta_0^*)_j$ is the majority value. 
 Furthermore, if most of the local parameter coordinates $(\theta_1^*)_j,\dots,(\theta_m^*)_j)$ are in a \emph{bulk} (considered as inliers) and a few of them are significantly different (outliers), then $(\theta_0^*)_j$ is the average of only the inliers. 
This ensures that the global parameter is \emph{sparse} and \emph{interpretable} as the common part of the local datasets.
\item \emph{statistical efficiency:} $\theta_0^*$ can be estimated at a faster rate than the $\theta_k^*$'s. This ensures meta-analysis leads to gains in statistical efficiency.
\end{enumerate}
We note that the ANOVA identification restriction does not share these two benefits. 

The rest of this paper is organized as follows. In Section \ref{sec:identifiability}, we describe the pitfalls of several common definitions of the global parameter in integrative studies and propose a robustness-based definition to avoid such pitfalls. We design an estimator of the robustness-based global parameter in Section \ref{sec:estimator} and prove in Section \ref{sec:theoreticalProperties} that it converges at a rate \emph{that depends on the total number of samples in all the datasets}. Finally, we demonstrate the benefits of our approach in predicting the response of rare cancers to therapeutics with the Cancer Cell Line Encyclopedia (CCLE).

\subsection{Related work}

The goal of meta-analysis is borrowing strength from different datasets, and the most common approach is a two-step method in which the local parameters are first estimated from their respective datasets and then combined with (say) a fixed-effects model \citep{hedges2014statistical} or a random-effects model  \citep{dersimonian1986meta}. In the high-dimensional setting, we must also perform variable selection to reduce the prediction/estimation error and ensure that the fitted model is interpretable.

\citet{gross2016Data,asiaee2018High} studied the case of heterogeneous linear regression models, where it is assumed that the underlying distribution of data sets are not the same. The former focuses on the prediction aspects of the problem, while the latter deals with the estimation aspects. Both papers reduce the problem of meta-analysis into a single lasso exercise, while the latter uses a version of the gradient descent algorithm to estimate parameters. Heterogeneity in the Cox model was studied by  \citet{cheng2015Identification} where the likelihood was maximized using a suitable lasso problem. However, all these methods required the full data sets for analysis to be available on a single platform. This raises the question of communication efficiency and privacy concerns pertaining to the data sets. \cite{cai2021individual} proposed a communication-efficient integrative analysis for high-dimensional heterogeneous data which addresses the issue of privacy preservation. In their two-step estimation procedure, they used lasso estimates and covariance matrices to obtain an estimator for the shared parameter, which, in a nutshell, is the average of the debiased lasso estimates.

There is a line of work on distributed statistical estimation and inference, \eg,  \citet{lee2017Communicationefficient, battey2018distributed, jordan2016CommunicationEfficient}, which is distinguished from our work by the additional assumption of no heterogeneity: \ie\ the datasets are identically distributed. In this line of work, there are two general approaches: averaging local estimates of the parameters \citep{lee2017Communicationefficient,battey2018distributed} and averaging local estimates of the score/sufficient statistic \citep{jordan2016CommunicationEfficient}. Although averaging the score has computational benefits over averaging the parameters, the latter is more amenable to meta-analysis because modeling the heterogeneity in the parameters is easier than modeling heterogeneity in the score.

\section{Identifying the global parameter}
\label{sec:identifiability}

We now quickly review the developments thus far. To distinguish between the common and idiosyncratic parts of the local parameters, we model the \emph{local parameters} $\theta_1^*,\dots,\theta_m^*$ with an additive model:
\[
\theta_k^* = \theta_0^* + \delta_k^*,\quad k\in[m].
\]
By itself, this additive model does not \emph{identify} the global parameter $\theta_0^*$ because the model is overparameterized. Thus it is necessary to impose additional identification restrictions to uniquely identify the global parameter. As we saw in Section \ref{sec:intro}, the choice of identification restriction is crucial to the statistical efficiency of meta-analysis. Recall we look for two properties in an identification restriction: interpretability and statistical efficiency. We describe them here again for the reader's convenience:
\begin{enumerate}
\item \emph{interpretability:} if there is a majority value among $(\theta_1^*)_j,\dots,(\theta_m^*)_j$ (i.e. a common value shared by more than half of $(\theta_1^*)_j,\dots,(\theta_m^*)_j$), then $(\theta_0^*)_j$ is the majority value. This encourages sparsity in the global parameter $\theta_0^*$ because if a majority of $(\theta_1^*)_j,\dots,(\theta_m^*)_j$ are zero, then $(\theta_0^*)_j$ is zero. 
\item \emph{statistical efficiency:} $\theta_0^*$ can be estimated at a faster rate than the $\theta_k^*$'s. This ensures meta-analysis leads to gains in statistical efficiency.
\end{enumerate}
In the rest of this section, we show that standard identification restrictions do not satisfy the two preceding properties. After describing the inadequacies of standard identification restrictions, we present two that satisfy our desiderata.

\subsection{Inadequacies of standard identification restrictions}

\paragraph{Mean:} The usual approach to borrowing strength across heterogeneous datasets is to model the variation among the local parameters with a prior and consider the (hyper)parameter of the prior as the global parameter. The celebrated empirical Bayes approach is a prominent example. One of the simplest and most widely used priors is the Gaussian prior: 
\[
(\theta_k^*)_j \sim N((\theta_0^*)_j,\sigma_0^2).
\]
The standard estimator of $(\theta_0^*)_j$ is $\frac1m\sum_{k=1}^m(\theta_k^*)_j$, which leads to the mean identification restriction \eqref{eq:anovaIdRestriction}. Unfortunately, this identification restriction does not satisfy our desiderata. Intuitively, the issue is it aggregates all local parameters in the definition of the global parameter regardless of whether a local parameter is similar to the other local parameters. This causes $(\theta_0^*)_j$ to lack interpretability because even a single non-zero local parameter is enough to nudge the global parameter away from zero. In the worst case (when the sparsity patterns of the local parameters are disjoint), the sparsity of the global parameter can be $m$ times the sparsity of the local parameters. This extra complexity of the global parameter negates any gains in statistical efficiency from meta analysis.

\begin{figure}
\centering
\includegraphics[width=0.75\linewidth]{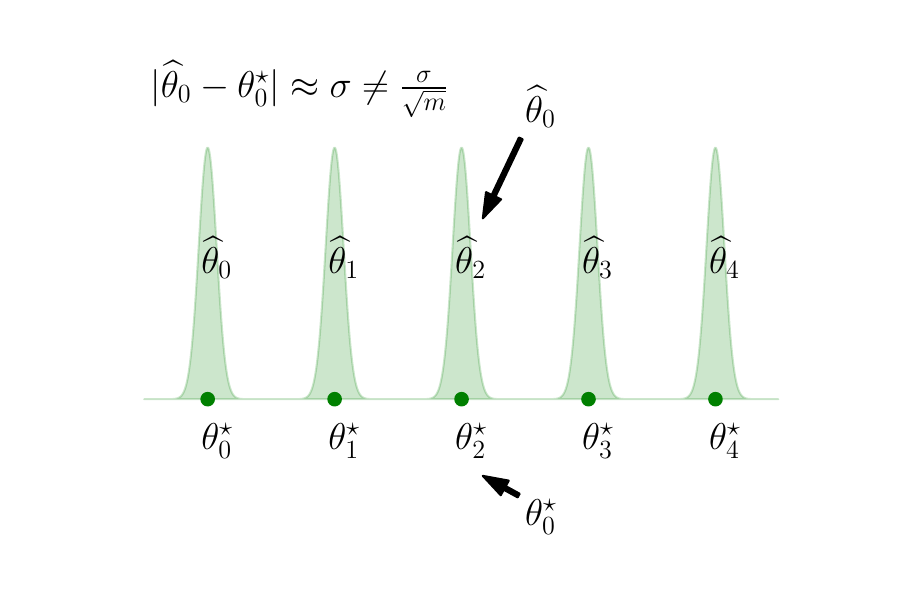}
\caption{Median fails to borrow strength form well-separated local parameter estimates. }
\label{fig:median}
\end{figure}

\paragraph{Median:} To address the sensitivity of the mean identification restriction to outliers, we consider identifying the global parameter with more robust centrality measures. One obvious choice is the median:
    \begin{equation}\textstyle
\theta_0^* \triangleq \arg\min_{\theta_0\in\reals^d}\sum_{k=1}^m\|\theta_k^* - \theta_0\|_1.
\label{eq:median}
\end{equation} Unfortunately, in certain scenarios, the median may \emph{fail to borrow strength} across datasets, making it statistically inefficient. If the local parameters are well-separated (see Figure \ref{fig:median}), then the global parameter is one of the local parameters. In this case, it is impossible to estimate the global parameter at a faster rate than the equivalent local parameter. We provide a detailed example in Appendix \ref{sec:counter-example} (see Example \ref{ex:medianFail}).

We note that \citet{asiaee2018High} and \citet{gross2016Data} both considered identification restrictions similar to \eqref{eq:median}. In particular, they studied the estimator 
\[
(\htheta_0, \hdelta_1, \ldots, \hdelta_m) \in \argmin_{(\theta_0,\delta_1, \dots , \delta_m)}\left\{\begin{aligned}&\textstyle\frac{1}{2N} \sum_{k=1}^m  \left\| \bY_{k} - \bX_{k} (\theta_0+ \delta_k)  \right\|_2^2 \\
&\textstyle\quad+ \frac1m(\lambda_0\|\theta_0\|_1 + \sum_{k=1}^m\lambda_k \|\delta_k\|_1)\end{aligned}\right\},
\] 
where $\bX_k\in\reals^{n_k\times d}$ and $\bY_k\in\reals^{n_k}$ are the features and responses in the $k$-th dataset and $N = \sum_{k = 1}^m n_k$ is the total sample size. This estimator implicitly enforces the identification restriction
\[\textstyle
\theta_0^* \triangleq \argmin_{\theta_0} \left(\sum_{ k =1}^m \|\thetas_k -  \theta_0 \|_1 + c\|\theta_0\|_1  \right)
\] 
for some $c > 0$, which, unfortunately, suffers from a similar drawback as \eqref{eq:median}. We refer to Result \ref{result:id-data-shared-lasso} in Appendix \ref{sec:counter-example} for the details. We note that this issue is reflected in \cite{asiaee2018High}'s theoretical results. In particular, they showed that 
\[\textstyle
\|\widehat\theta_0 - \theta_0^*\|_2 + \sum_{k=1}^m\frac{1}{\sqrt{m}}\|\widehat\delta_k - \delta_k^*\|_2 \lesssim_P \max_km(\frac{\max_ks_k\log d}{N})^{\frac12},
\]
where $s_k$ is the sparsity of $\theta_k^*$. If we assume $\theta_0^*$, $\delta_1^*,\dots,\delta^*_m$ are all $s$-sparse, then the right side simplifies to $\sqrt{m}(\frac{s\log d}{n})^{\frac12}$. While this is the fastest rate of convergence we expect for the left side, it suggests that the estimator of the global parameter $\htheta_0$ converges at a rate that depends on the local sample size. 

\paragraph{Huber loss:} As an alternative to the median, we consider minimizing Huber's loss \citep{huber1964Robust}. Huber's loss function is defined as  \[ L_\eta (x)  = \begin{cases} \frac12 x^2 & \text{if } |x|\le \eta \\
\eta \left(|x| - \frac12 \eta \right) & \text{otherwise},
\end{cases}\]
for some $\eta > 0$, and the global parameter is defined (coordinate-wise) as:
\begin{equation}\textstyle
(\theta_0^*)_j \triangleq \argmin_{\theta_0} \sum_{k=1}^m L_\eta((\theta_k^*)_j - \theta_0)
    \label{eq:huber-id}
\end{equation}
We hope that the quadratic part of Huber's loss allows it to borrow strength effectively across datasets. Unfortunately, this leads to a loss of interpretability. Consider a problem in which all but one local parameter values are identically zero. It is possible to show that the global parameter is non-zero for any $\eta>0$ (see Example \ref{ex:huber-loss-counter} in Appendix \ref{sec:counter-example}).

\subsection{Two interpretable and statistically efficient global parameters}

The three examples presented in this subsection are hardly exhaustive. We list them here to illustrate the delicacy of the choice of identification restriction. In the rest of this section, we provide two examples of identification restrictions that satisfy our desiderata. 

Both the examples have hyperparameters, that one must choose carefully, depending on the properties of the local datasets to satisfy our desiderata. At first, this seems unusual because the global parameter depends on the hyperparameters. However, the global parameter is defined by the data scientist to facilitate meta-analysis, and the hyperparameters \emph{reflect the discretion of the data scientist}. Thus the dependence of the global parameter on the hyperparameters is natural. From another perspective, different values of hyperparameter lead to possible global parameters. However, not all the global parameters are interpretable and neither can they be estimated in a statistically efficient manner. We only consider  global parameters that have these desirable properties.

\paragraph{Re-descending loss:} In this paper, we define the global parameter as the minimum of a re-descending loss function \citep{huber1964Robust}. More concretely, we consider the quadratic re-descending loss function: $\Psi_\eta(x) \triangleq \min \{x^2,\eta^2\}$, and we define the global parameter as
\begin{equation}\textstyle
 	(\theta_0^*)_j \triangleq \argmin_{x\in\reals} \sum_{k= 1}^m \Psi_{\eta_j} ((\theta_{k}^*)_j-x)\,.
 	\label{eq:redescendindIDRestriction}
 \end{equation}
It is possible to derive similar results for other re-descending loss functions, but we focus on the quadratic re-descending loss function here. As we shall see, $(\theta_0^*)_j$ is not only interpretable but also statistically efficient.

First, we check that \eqref{eq:redescendindIDRestriction} leads to an interpretable global parameter. The quadratic re-descending loss function has the property that its derivative vanishes outside $[-\eta_j,\eta_j]$, so \eqref{eq:redescendindIDRestriction} ignores any local parameters that are outside this interval. Thus local parameters that are far from the bulk of the local parameter values are considered outliers and ignored by \eqref{eq:redescendindIDRestriction}, thereby ensuring that the global parameter is interpretable. For example, if there is a majority value among $(\theta_1^*)_j,\dots,(\theta_m^*)_j$, then for a suitable choice of $\eta_j$, \eqref{eq:redescendindIDRestriction} ignores all local parameters that are different from the majority value. 

Second, we argue that it is possible to estimate $(\theta_0^*)_j$ at a fast rate. Consider $(\theta_0^*)_j$ as an $M$-estimator in which the effective sample size is the number of local parameters in bulk (not considered outliers). This suggests that it is possible to estimate $(\theta_0^*)_j$ at a rate that depends on the number of local parameters in the bulk. As we shall see, the simple approach of replacing the local parameters with their estimates in \eqref{eq:redescendindIDRestriction} leads to an estimator that achieves this goal. We refer to Lemma \ref{lemma:integrated-hp-bound} for a formal statement of a result to this effect.

We note that the choice of $ \eta_j$ plays a crucial role in the interpretability and statistical efficiency of the global parameter $(\theta_0^*)_j$. If $ \eta_j $ is very large, then none of the local parameter values will be identified as outliers and \eqref{eq:redescendindIDRestriction} is equivalent to the mean identification restriction. This leads to a global parameter that is sensitive to outlier local parameters. On the other hand, if $ \eta_j $ is small, then \eqref{eq:redescendindIDRestriction} ignores most local parameter values because it considers them as outliers. This reduces the statistical efficiency of borrowing strength across datasets because most datasets are ignored. Thus $\eta_j$ must be chosen in a way that balances interpretability and statistical efficiency.

We note that the quadratic re-descending loss has a close connection with the $\eps$-contaminated model \eqref{eq:gen-contamination}, where the inliers parameters are normally distributed:
\begin{equation}
    (\theta_k^*)_j \sim (1-\eps) \cN\big((\theta_0^*)_j, \sigma^2\big) + \eps G_j\,.
\end{equation} In the absence of contamination ($\eps = 0$), it is not hard to check that minimizing the squared loss function leads to the maximum likelihood estimator for the global parameter $\theta_0^*$. In the presence of outliers, we can avoid the corrupting effects of the outliers by using a robust version of the squared loss function. There are many choices of such robust loss; one such choice is a re-descending loss function \citep{hampel2005Robust}. 

\paragraph{Quadratic + $\ell_1$ loss} As an alternative to the re-descending loss, we could also consider minimizing a convex combination of quadratic and $\ell_1$ loss functions:
\begin{equation}\textstyle
\theta_0^* \triangleq \arg\min_{\theta\in\reals}\sum_{k=1}^m \frac{1}{1+\lambda}\big\{\lambda\|\theta_k^* - \theta\|_1 + \frac12\|\theta_k^* - \theta\|_2^2\big\}\,.
\label{eq:sq-abs}
\end{equation} 
This identification restriction combines the mean and the median identification restrictions, but in a different way than Huber's loss function. It is not hard to check that this identification restriction satisfies our desiderata. The quadratic part of \eqref{eq:sq-abs} allows us to borrow strength across datasets, while it is possible to pick $\lambda$ in a way such that $\theta_0^*$ is interpretable. For example, if all but one of the local parameters are zero, it is possible to pick $\lambda$ large enough so that the global parameter is zero, in contrast to the Huber loss function. That said, we focus on the re-descending loss here because it has better empirical performance (see Section \ref{sec:computationalResults}).

\section{Communication-efficient data enriched regression}
\label{sec:estimator}
In this section, we suggest a privacy-preserving communication-efficient estimator for the global parameter $ \theta^*_0 $  which borrows strength over different datasets. The privacy concern limits us to communicate with the datasets only through some summary statistics. So, the high-level idea to estimate $ \theta_0^* $ is to start with some estimator of the local parameters computed from the datasets. Then the global parameter is estimated only using local estimates, without any further communication among datasets.

We describe the debiased lasso estimator for local parameters in the set-up of  $\ell_1$ regularized M-estimators. The case of the linear regression model can be considered as a special case.
 Let $\rho(y,a)$ be a loss function, which is convex in $a$, and $\drho$, $\ddrho$ be its derivatives with respect to $a$. That is
\[
\drho(y,a) = \frac{d}{da}\rho(y,a),\quad\ddrho(y,a) = \frac{d^2}{da^2}\rho(y,a).
\] 
Define $\ell_k(\theta_k) = \frac1{n_k}\sum_{i=1}^{n_k}\rho(\by_{ki},\bx_{ki}^T\theta_k)$, where the sum is only over the pairs on dataset $k$. The lasso estimator for local parameter $ \theta_{k}^* $ is given by

\begin{equation}
	\ttheta_{k} \coloneqq\argmin_{\theta\in\reals^d}\ell_k(\theta) + \lambda_k\|\theta\|_1.
	\label{eq:lasso-m-estimator}
\end{equation}
Since averaging only reduces variance, not bias, we gain (almost) nothing by averaging the biased lasso estimators. That is, the MSE of the naive average estimator is of the same order as that of the local estimators. The key to overcoming the bias of the averaged lasso estimator is to ``debias'' the lasso estimators before averaging.

The \emph{debiased lasso estimator} as in  \citet{vandegeer2014asymptotically} is 
\begin{equation}
\ttheta_{k} ^d = \ttheta_{k} - \Thetah_k\nabla\ell_k(\ttheta_k),
\label{eq:debiased-lasso-m-estimator}
\end{equation}
where $\Thetah_k$ is an approximate inverse of $\nabla^2\ell_k(\ttheta_k)$. The choice of  $ \Thetah_k $ in the correction term crucial to the performance of the debiased estimator. In particuar, $\Thetah_k$ must satisfy
\[
\|I - \Thetah_k\nabla^2\ell_k(\ttheta_k)\|_\infty \lesssim \left(\frac{\log d}{n_k}\right)^{\frac12}.
\]

One possible approach to forming $\Thetah_k$, as in \citet{vandegeer2014asymptotically},  is by nodewise regression on the weighted design matrix $\bX_{\widehat\theta_k} := W_{\widehat\theta_k}\bX_k$, where, $\bX_k $ is the $n_k \times d$ design matrix for $k$-th dataset, defined as $ \bX_k = \BMAT \bx_{k1} & \bx_{k2} & \dots & \bx_{kn_k} \EMAT^\top ,$ and   $W_{\widehat\theta_k}$ is $n_k \times n_k$ diagonal matrix, whose diagonal entries are
\[
\bigl(W_{\widehat\theta_k}\bigr)_{i,i} := \ddrho(y_{ki},x_{ki}^T\widehat\theta_{k})^{\frac12}.
\]
That is, for $j \in [p]$ that machine $k$ is debiasing, the machine solves
\begin{equation}
\hgamma^{(k)}_j := \argmin_{\gamma\in\reals^{p-1}}\frac{1}{2n_k}\|\bX_{\widehat\theta_k,j} - \bX_{\widehat\theta_k,-j}\gamma\|_2^2 + \lambda_j\|\gamma\|_1,\,j\in[p],
	\label{eq:gamma}
\end{equation}
and forms
\begin{equation}
\left(\Thetah_k\right)_{j,\cdot} = \frac{1}{\htau_{kj}^2}\BMAT-\hgamma^{(k)}_{j,1} & \dots & -\hgamma^{(k)}_{j,j-1} & 1 & -\hgamma^{(k)}_{j,j+1} & \dots &-\hgamma_{j,p}^{(k)}\EMAT,
\label{eq:theta-hat-glm}
\end{equation}
where
\[
\htau_{kj} = \left(\frac1{n_k}\|\bX_{\widehat\theta_k,j} - \bX_{\widehat\theta_k,-j}\hgamma_j\|_2^2 + \lambda_j\|\hgamma_j\|_1\right)^\frac12.
\]

After calculating the \textit{debiased lasso estimators} $ \tthetadk$ in local datasets, the integrated estimator is obtained using similar optimization problem as in \eqref{eq:redescendindIDRestriction}.  $ j $-th co-ordinate of the \emph{integrated debiased lasso} estimator  is obtained as
\begin{equation}
	\left(\ttheta_0\right)_j = \argmin_{x} \sum_{k=1}^{m}  \Psi_{\eta_j}\left(\left(\ttheta^d_k\right)_j-x\right),
	\label{eq:averaged-debiased-lasso}
\end{equation}
where, $ \eta_j $'s used in the above equation are the same ones that are used to identify the global parameter. 
The integrated estimator $ \tilde\theta_0 $  calculated in \eqref{eq:averaged-debiased-lasso} has a serious drawback. Loosely speaking, the use of quadratic re-descending loss to define  $ \ttheta_0 $ gives us the co-ordinate wise mean of debiased lasso estimates, after dropping the outlier values.  Since the debiased estimates are no longer sparse, $ \ttheta_0 $ is also dense. This detracts from the interpretability of the coefficients and makes the estimation error large in the  $\ell_2$ and $\ell_1$ norms. To remedy both problems, we threshold $ \ttheta_0. $ Below we describe the hard-threshold and soft-threshold on $ \ttheta_0 $:
\begin{align}
\label{def:st-ht}
HT_t\left(\ttheta_0\right) &\gets \left(\ttheta_0\right)_j \cdot \ones_{\left\{\abs{\left(\ttheta_0\right)_j} \ge t\right\}}, \text{or}\\\notag
ST_t\left(\ttheta_0\right) &\gets \sign\left( \left(\ttheta_0\right)_j\right)\cdot  \max\left\{\left| \left(\ttheta_0\right)_j\right|-t,\,0\right\}.
\end{align}

The final estimator  $\widehat\theta_0$ is obtained by hard-thresholding or soft-thresholding  $\ttheta_0$ at $t\sim \sqrt{\frac{\log d }{N_{\text{min}}}},$ where, $ N_{\min} = m \min_k n_k, $ \ie, \ $ \widehat\theta_0 = HT_t\left(\ttheta_0\right)  \text{ or }  ST_t\left(\ttheta_0\right), $ and the final estimators for $\deltavec_k$'s are obtained by hard-thresholding or soft-thresholding $\ttheta_k^d-\ttheta_0 $ at a level $t'\sim \sqrt{\frac{\log d }{n_k}},$ \ie, \ $\widehat\delta_k \gets HT_{t'}\left( \tdelta_k \right) $ or $ ST_{t'}\left( \tdelta_k \right) .$ 

A step by step process to obtain the  global estimator $ \widehat\theta_0 $ for linear regression models is described in Algorithm \ref{Alg:CommEffEst}. 
The Algorithm requires suitable choices of the parameters $\{\eta_j\}$, $t$ and $\{t_k\}$. In our numerical studies we pick these choices via the cross-validation approach, described in Algorithm \ref{Alg:cv}.  
If appropriate choices of $\eta_j$'s are available from background knowledge then one can use those instead of picking them from cross-validation. In practice,  cross-validating over $(\{\eta_j\}, t, \{t_k\})\big)$ might be computationally challenging since there are $d + m +1$ many unknown parameters for $m$ local datasets and   $d$ as the covariate dimension; and specifically $d$ can be quite large. In our experiments we further simplify the choice by assuming $\eta_j$'s have equal value $\eta$ and $t_k = t \sqrt{N_{\min}/n_k}$. This simplifies the cross-validation parameters  $(\{\eta_j\}, t, \{t_k\})\big)$ to just two parameters $(\eta, t)$.

\begin{algorithm}

	\DontPrintSemicolon
	
	\textbf{Input:} $\left\{\cD_k = (\{\bx_{ki},\by_{ki}\})_{i=1}^{n_k} :k\in[m]\right\},\  \{ \eta_j \}_{j \in[d]}$.\\
	\textbf{Output:} $\left(\widehat{\theta_0},\widehat\delta_1,\ldots,\widehat\delta_m\right)$.\\
	\For{$ k = 1,2,\cdots,m, $}{
	$ \ttheta_k \gets \argmin_{\theta\in\reals^d}\frac{1}{2n_k}\sum_{i=1}^{n_k} \rho (\by_{ki},\bx_{ki}^T\theta) + \lambda_k\|\theta\|_1 $
	
	$ 	\ttheta^d_k \gets \ttheta_{k} - \Thetah_k\nabla\ell_k(\ttheta_k), $ where, $ \Thetah_k $ are defined as in \ref{eq:theta-hat-glm},
}

	$\left(\tilde{\theta_0}\right)_j\gets \argmin_{x} \sum_{k=1}^{m} \Psi_{\eta_j}\left(\left(\ttheta^d_k\right)_j-x\right), $ 
	
	$\widehat{\theta}_0\gets HT_t\left(\ttheta_0\right) \text{ or } ST_t\left(\ttheta_0\right)$ for $t\sim \sqrt{\frac{\log d }{N_{\min}}},$ where, $ N_{\min} = m\min_k n_k, $

	\For{$k\in [m],$}
	{
		$\tdelta_k \gets \ttheta^d_k - \ttheta_0,$
		
		$ \widehat\delta_k \gets HT_{t_k}\left( \tdelta_k \right) $ or $ ST_{t_k}\left( \tdelta_k \right) $ for $ t_k\sim \sqrt{\frac{\log d }{n_k}}. $
		
	}

	\caption{MrLasso($\{\eta_j\}, t, \{t_k\}$)}
		\label{Alg:CommEffEst}
\end{algorithm}

\begin{algorithm}

	\DontPrintSemicolon
	
	\textbf{Input:} $\cD = \left\{\cD_k = (\{\bx_{ki},\by_{ki}\})_{i=1}^{n_k} :k\in[m]\right\},\ C$, $G$ the grid of ($\{\eta_j\}$, $t$, $\{t_k\}$).\\
	\textbf{Output:} $\left(\widehat{\theta_0},\widehat\delta_1,\ldots,\widehat\delta_m\right)$.\\
	
	$\{\cD^{(c)}\}_{c = 1}^C \gets C \text{ equal partition of } \cD$ 
	
	\For{$(\{\eta_j\}, t, \{t_k\}) \in G$}{
	
	\For{$c = 1, \dots, C$}
	{
	$\cD^{(-c)}\gets \cD \backslash \cD^{(c)}$
	
	$\left(\widehat{\theta_0}^{(c)},\widehat\delta_1^{(c)},\ldots,\widehat\delta_m^{(c)}\right) \gets \text{MrLasso}\big(\cD^{(-c)}, (\{\eta_j\}, t, \{t_k\})\big) $

    $\cE_c \gets \frac 1{|\cD^{(c)}|}\sum_{k = 1}^m \sum_{(x, y)\in \cD^{(c)}_k}\rho \{y, x^\top (\widehat{\theta_0}^{(c)} + \widehat\delta_k^{(c)})\}$

	}
$\cE(\{\eta_j\}, t, \{t_k\}) \gets \frac 1C \sum_{c = 1}^C \cE_c$	
	
}
	
	$(\{\hat \eta_j\}, \hat t, \{\hat t_k\}) \gets \argmin ~ \cE(\{\eta_j\}, t, \{t_k\})$

	$\left(\widehat{\theta_0},\widehat\delta_1,\ldots,\widehat\delta_m\right) \gets \text{MrLasso}\big(\cD, (\{\hat \eta_j\}, \hat t, \{\hat t_k\})\big) $

	\caption{Cross-validation}
		\label{Alg:cv}
\end{algorithm}

\section{Theoretical properties of the communication-efficient estimator}

\label{sec:theoreticalProperties}
To present the theoretical justification of consistency of the estimators we shall focus on $\ell_1,\ell_2,$ and $\ell_{\infty}$ consistency of the estimators. Before getting into the assumptions and results, we first define some notations. 	We use $\|\cdot\|_1$, $\|\cdot\|_2$, and $\|\cdot\|_\infty$ to denote usual $\ell_1$-norm, $\ell_2$-norm,
and $\ell_\infty$-norm respectively.

The performance of global estimator depends on the debiased lasso estimators from the local datasets. Hence, it is important to have a reasonable performance for local estimators. We study the $ \ell_\infty $ error rate of the debiased lasso estimator $ \ttheta^d $ for the parameter $ \theta, $ calculated form  the dataset $ \{\bx_i,\by_i\}_{i=1}^n. $ For a random vector $ (\by,\bx) $, whose distribution is parametrized by $ \theta, $ we assume that  $ \Ex [\rho(\by,\bx^T\theta)] $ is uniquely minimized at $ \theta. $ As before, the debiased lasso estimator is defined as $ \ttheta^d \triangleq \ttheta - \Thetah \nabla \ell (\ttheta)  ,$ where, $ \ttheta = \argmin_{\theta\in\reals^d} \ell(\theta ) + \lambda\|\theta\|_1, $ for $ \ell(\theta ) = \frac{1}{n} \sum_{ i= 1}^n \rho(\by_i,\bx_i^T\theta ), $ and $ \Thetah  $ is calculated according to \eqref{eq:theta-hat-glm}. We assume that the actual parameter value  $ \theta^* $ is $ s_0 $ sparse, \ie,  $ \theta^* $ has $ s_0 $ many non-zero entries. Letting $ \Theta \triangleq (\Ex \nabla^2 \ell(\thetas))^{-1} $ we assume the $j  $-th row of $ \Theta $, which is denoted as $ \Theta_{j,\cdot} $ is $ s_j $-sparse. Denote $s^* = \max_j s_j$.   We make the following assumptions on distributions of local datasets. These assumptions establish high probability bounds for the local debiased lasso estimates.

\begin{assumption}
\label{assump:debiased-lasso-assumption}
    Under the M-estimation setup, we assume the following. 
    \begin{enumerate}
        \item The pairs  $\{(x_i, y_i)\}_{i = 1}^n$ are \iid\  $P$ and for some $K\ge 1,$  $\|\bX\|_{\infty} = \max_{i, j} |X_{i,j}| = \cO(K)$ and the projection of $ \bX_{\thetas, j} $ on the row space of $ \bX_{\thetas, -j} $ in the $ \Sigma_\thetas \triangleq \Ex [\nabla^2 \ell(\thetas)] $ inner product  is bounded: $  \|\bX_{\thetas, -j}\gamma_{j}\|_\infty = \cO(K) $ for any $ j, $ where \[ \gamma_{j} \triangleq  \argmin_{\gamma\in\reals^{p-1}} \Ex\left[ \|\bX_{\thetas, j}- \bX_{\thetas, -j}\gamma\|_2^2 \right]. \] We define $ \tau_j^2 $ as $ \tau_j^2 \triangleq \Ex \big[\frac1n \| \bX_{\thetas, j}-\bX_{\thetas, -j}\gamma_j\|_2^2\big]. $
        
        \item It holds that $K^2 s_j \sqrt{\log d/n} = o(1).$
         \item The smallest eigenvalue of $ \Sigma_{\thetas} $ is bounded away from zero, and moreover, $\|\Sigma_{\thetas}\|_{\infty } = \cO(1).$
        \item There exists a  $\delta>0$ such that  for any  $\theta$ that satisfies  $ \|\theta-\thetas\|_1\le \delta $ it holds that  $ \ddrho (\by_i,\bx_i^T\theta) $ stay away form zero and that $\|\ddrho (\by_i,\bx_i^T\theta)\|_{\infty} = \cO(1)$. Furthermore, for any such $\theta$ and all $x, y$ \[ |\ddrho(y,x^T\theta)-\ddrho(y,x^T\thetas)|\le |x^T(\theta-\thetas)|. \]
        	\item With probability at least $ 1-o(d^{-1}) $ we have $ \frac1n \|\bX(\ttheta-\thetas)\|_2^2 \lesssim s_0 \frac{\log d}{n}  $ and $ \|\ttheta - \thetas\|_1 \lesssim s_0\sqrt{\frac{\log d}{n}} $.
		\item The derivative $  \drho(y,a),\ \ddrho(y,a) $ exists for all $ y,a, $ and for some $ \delta $-neighborhood, $ \ddrho(y,a) $ is locally Lipschitz: for some $L> 0$ \[ \max_{a_0\in\{x_i^T\thetas\}} \sup_{|a-a_0|\vee |\widehat a - a_0|\le \delta} \sup_{y\in\cY} \frac{|\ddrho(y,a)-\ddrho(y,\widehat a)|}{|a-\widehat a|} \le L. \] Moreover,
		\[ \max_{a_0\in\{x_i^T\thetas\}} \sup_{y\in\cY} |\drho(y,a_0)| = \cO(1), \ \max_{a_0\in\{x_i^T\thetas\}} \sup_{|a-a_0|\le \delta} \sup_{y\in\cY} |\ddrho(y,a_0)| = \cO(1). \]
		\item The diagonal entries of $ \Theta \Ex \left[ \nabla\ell(\thetas) \nabla\ell(\thetas)^T \right]\Theta $ are bounded by $\sigma^2$.
    \end{enumerate}
\end{assumption}
The list of conditions in Assumption \ref{assump:debiased-lasso-assumption} are standard ones to establish $\ell_\infty$ convergence for the debiased lasso estimator and are adapted from \cite[Assumptions (D1)-(D5)]{vandegeer2014asymptotically} and \cite[Assumptions (B1)-(B7)]{lee2017Communicationefficient}.   
Condition 5 that directly assumes convergence of estimation and prediction error for lasso is implied by the other assumptions.  We refer to \cite[Chapter 6, Section 6.7]{buhlmann2011statistics} for the details, where the necessary compatibility condition is inherited from the condition 3. 
Here we state this to simplify the exposition. We refer to the corresponding literature (specifically \cite{lee2017Communicationefficient}) for further discussions about the conditions.

In the literature on debiased lasso for high-dimensional  M-estimation, it is usually assumed that $K = \cO(1)$ (see \cite[Section 3.3.1]{vandegeer2014asymptotically}, \cite[Section 5]{lee2017Communicationefficient}). For sub-gaussian covariates one can use a high-probability bound of the order $K = \cO_\Pr(\sqrt{\log(dn)})$. 
The high probability $\ell_\infty$ bound for local debiased lasso estimate follows.

\begin{lemma}
	 Under Assumption \ref{assump:debiased-lasso-assumption}  the debiased lasso estimator $ \ttheta^d $ satisfies the following high probability bound.
	 \begin{equation}
	 	\text{for some }c>0,\  \Pr\left(\|\ttheta^d - \theta^*\|_\infty > \sigma \sqrt{\frac{12\log d}{n}}  + c\frac{ \max\{Ks^*, K^2 s_0\}\log d}{n}\right) \le o(d^{-1}).
	 	\label{eq:concentration-local-estimator}
	 \end{equation}
	 \label{lemma:local-estimate-tail-bound}
\end{lemma}

\begin{remark}
	Under the Assumption \ref{assump:debiased-lasso-assumption}  and  $K^2s_0\sqrt{\frac{\log d}{n}} = o(1)$ we have a high probability bound $ \|\ttheta^d-\thetas\|_\infty \lesssim_\Pr \sqrt{\log d/n}. $ 
	\label{remark:hp-bound-local} 
\end{remark}

Before studying the theoretical properties of the integrated estimator defined in \eqref{eq:averaged-debiased-lasso} one might notice that the identification restriction \eqref{eq:redescendindIDRestriction} may not lead to a unique global parameter, in general. The whole business of studying the theoretical properties of the global estimator is not meaningful if one cannot uniquely identify the global parameter. Hence, it is necessary to make some assumptions about the parameter values, with the goal that they will suffice unique identification of the global parameter. Here, we present the assumptions on the finite number of machines to have a uniquely identified parameter. 
For $ j\le d $ we assume
\begin{figure}
    \centering
    \includegraphics[width=\linewidth]{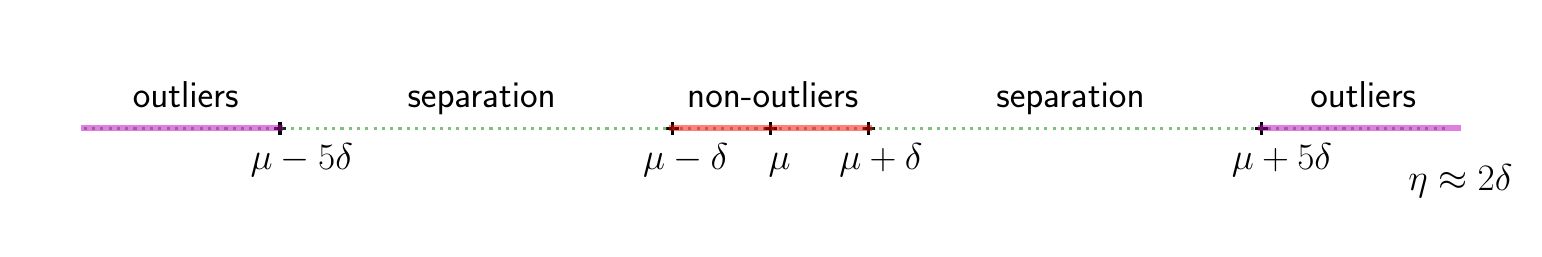}
    \caption{Bulk structure for the local parameters (Assumption \ref{assump:cluster-parameter}). This structure ensures the global parameter \eqref{eq:redescendindIDRestriction} is unique.}
    \label{fig:assump}
\end{figure}

\begin{assumption}
	\begin{enumerate}[label=(B\arabic*)]
		\item 
		Let $ I_j $ be the set of indices for $ (\thetas_{k})_j $'s which are considered as inliers. We assume $|I_j|/m \ge 4/7$.
		\item  Let $ \mu_j = \frac1{|I_j|}\sum_{ k \in  I_j}(\thetas_{k})_j . $ Let $  \delta  $ be the smallest positive real number such that $ (\thetas_k)_j \in [\mu_j-\delta,\mu_j+\delta] $ for all $ k\in I_j. $
		 We assume that none of the  $ (\thetas_k)_j $'s are in the intervals $ [\mu_j-5\delta,\mu_j-\delta) $ or $ (\mu_j+\delta,\mu_j+5\delta]. $ 
		\item Let $ \delta_2 = \min_{k_1\in I_j, k_2\notin I_j} |(\thetas_{k_1})_j - (\thetas_{k_2})_j| $ is the minimum separation between inliers and outliers.  Clearly, $ 4\delta<\delta_2. $ We choose $\eta_j$ such that $ 2\delta < \eta_j <\delta_2/2. $
	
	\end{enumerate}
\label{assump:cluster-parameter}
\end{assumption}

The Assumption \ref{assump:cluster-parameter} sets a co-ordinate wise clustering assumption for the local parameters. A visualization summary of the cluster Assumption \ref{assump:cluster-parameter} is seen in Figure \ref{fig:assump}. The conditions ensure enough separation between the bulk and the outliers  such that the re-descending loss \eqref{eq:redescendindIDRestriction} uniquely identifies the global parameter at mean of the bulk values, as suggested by Result \ref{result:global-uniqueness}. 
\begin{result}
	Under the Assumption \ref{assump:cluster-parameter} the objective function $ x \mapsto\sum_{ k= 1}^m \Psi_{\eta_j}\left((\thetas_{k})_j - x \right) $ has a unique minimizer $ (\thetas_0)_j = \mu_j. $
	\label{result:global-uniqueness}
\end{result}
Result \ref{result:global-uniqueness} implies that under the Assumption \ref{assump:cluster-parameter} the global parameter $ \thetas_0 $ is uniquely identified at the mean of the inlier parameter values. This has
the following nice implications: (1) the global parameter is not affected by outlier values; (2) the estimator of the global parameter can effectively borrow strength across datasets with inlier parameter values. The Lemma \ref{lemma:integrated-hp-bound}, where we study the coordinate-wise convergence rates for $ \ttheta_{0}$, is a step toward establishing effective borrowing strength for the estimator of global parameter.

\begin{lemma}
     Let the following hold:
\begin{enumerate}[label=(\roman*)]
	\item For any $ j, $ $ \left\{(\thetas_{k})_j\right\}_{k=1}^n $ satisfy the Assumption \ref{assump:cluster-parameter}.
	\item The datasets $ \left\{\cD_k\right\}_{k=1}^m $ satisfy  the Assumption \ref{assump:debiased-lasso-assumption} uniformly over $ k $.  
	\item Let $s_{k, 0}$ be the sparsity for $\theta_k^*$ and $s_{k}^*$ being the maximum sparsity for the rows of $\Theta_k \triangleq \{\Ex[\nabla^2 \ell_k(\theta^*_k)]\}^{-1}$. Define $s_{0, \max} = \max_k s_{k, 0}$,  $s_{\max}^* = \max_k s_k^*$ and $s_{\operatorname{eff}} = \max\{K s_{\max}^*, K^2 s_{0,\max}\}$.  We assume $m = o\left( \frac{n^2_{\min}}{s_{\operatorname{eff}}^2n_{\max}\log d }\right),$ where, $ n_{\max} = \max_{k \in [m]} n_k, $ and $n_{\min} = \min_{k \in [m]} n_k.$ 
\end{enumerate}

Then for   sufficiently large $ n_k $ we have the following bound for the co-ordinates of $ \ttheta_0 $ and $ \ell_\infty $ bound for $ \tdelta_k $:
\begin{equation}
	 \left| (\ttheta_0)_j - (\theta_0^*)_j \right| \le  4\sigma\sqrt{\frac{\log d}{|I_j|^2}\sum_{ k \in  I_j}\frac1{n_k}}, \text{ for all }j,\text{ and } \| \tdelta_k - \deltas_k \|_\infty  \le  4\sigma\sqrt{\frac{\log d}{n_k}}, \text{ for all }k,
	\label{eq:bound-co-ord-glob}
\end{equation} with probability at least $ 1-o(1). $
\label{lemma:integrated-hp-bound}
\end{lemma}

Condition $(i)$ in lemma \ref{lemma:integrated-hp-bound} is needed to identify the global parameter uniquely, as suggested by Result  \ref{result:global-uniqueness}. Condition $(ii)$ provides a high probability $\ell_\infty$ bound for the local debiased lasso estimates (see lemma \ref{lemma:local-estimate-tail-bound}). Finally, the condition $(iii)$ ensures the bias in local debiased estimates is of smaller order than the variance in integrated estimate. Since the bias term may not improve under averaging, violating this condition would result in higher-order for the bias in local estimates, and the estimation error for the global parameter stops improving (with a higher number of datasets).

In the following remark we introduce a weighted version of our data integration method \eqref{eq:averaged-debiased-lasso}.
\begin{remark}[A weighted integration]

One may also consider \emph{a weighted version} of the data integration \begin{equation}
	\left(\ttheta_0^{w}\right)_j = \argmin_{x} \sum_{k=1}^{m} w_k  \Psi_{\eta_j}\left(\left(\ttheta^d_k\right)_j-x\right),
	\label{eq:averaged-debiased-lasso-weighted}
\end{equation} where $w_k\ge 0$ is the weight corresponding to $k$-th dataset with $\sum_{k = 1}^m w_k = 1$. In that case the uniqueness of the corresponding global parameter \begin{equation}
	\left(\theta_0^{w, *}\right)_j = \argmin_{x} \sum_{k=1}^{m} w_k  \Psi_{\eta_j}\left(\left(\theta^*_k\right)_j-x\right),
	\label{eq:weighted-global-mrlasso}
\end{equation} at the value $\left(\theta_0^{w, *}\right)_j = \sum_{k \in I_j} w_k (\theta^*_k)_j / (\sum_{k\in I_j} w_k)$ is established (in Result \ref{result:global-uniqueness-nk}) by considering a weighted version of Assumption \ref{assump:cluster-parameter}  (see Assumption \ref{assump:cluster-parameter-nk}) where  we replace $|I_j|/m$ by $ \sum_{k\in I_j} w_k$ in condition (B1), and let $\mu_j = \sum_{k \in I_j} w_k (\theta^*_k)_j / (\sum_{k\in I_j} w_k)$ in condition (B2). Furthermore, under conditions (ii) and (iii) in Lemma \ref{lemma:integrated-hp-bound}, one can  establish (see Theorem \ref{th:integrated-hp-bound-nk}) the following high probability 
bound for $\left(\ttheta_0^{w}\right)_j$:
\[
\bbP \left[ \left| (\ttheta_0^w)_j - (\theta_0^{w, *})_j \right| \le 4 \sqrt{\frac{\log d}{\big(\sum_{k \in I_j }w_k\big)^2}\sum_{ k \in  I_j}\frac{w_k^2}{n_k}}, \text{ for all }j\right]  = 1 - o(1)\,.
\] Further details about the setup are provided in Appendix \ref{sec:weighted-mrlasso}. 

An interesting special case arises when $w_k = n_k/N$, and the above rate   simplifies to $\cO\big( \{\sum_{k \in I_j}  n_k\}^{-1}\big)$ and  yields the fastest rate among all possible weights: 
\[
\frac 1{\sum_{k \in I_j}  n_k} = \min_{w} \Big[\frac{1}{\big(\sum_{k \in I_j }w_k\big)^2}\sum_{ k \in  I_j}\frac{w_k^2}{n_k}\Big]\,.
\] Note that the rate $\cO\big( \{\sum_{k \in I_j}  n_k\}^{-1}\big)$ is faster than the corresponding one in \eqref{eq:bound-co-ord-glob}, especially when the $n_k$'s are significantly different. The downside to the weighted formulation is that the corresponding structural assumptions on the local parameter values depends on the sample size, which can be perceived as strange.

\end{remark}

\begin{proof}
	Here we shall provide a brief outline of the proof of Lemma \ref{lemma:integrated-hp-bound}. The detailed proof of this lemma is provided in appendix  \ref{subsec:supp-results}. We start with the Remark \ref{remark:hp-bound-local}, which gives us a high probability error bound for the individual local lasso estimators. From condition $ (i) $ we see that the $\thetas_j$'s satisfies Assumption \ref{assump:cluster-parameter}. As long as $ \ttheta_k^d $ concentrates around $ \thetas_k $, $\ttheta_k^d$ also satisfies Assumption \ref{assump:cluster-parameter} (with different $ \delta_j $ and $\delta_{2,j}$'s but identical $ \eta_j $'s).
	Using Result \ref{result:global-uniqueness} we see that the integrated estimator has the form
	\[ (\ttheta_0)_j = \frac{\sum_{ k \in  I_j} (\ttheta_k^d)_j}{|I_j|}.  \] Finally, we use the decomposition on the individual lasso estimators $ \tilde\theta_k^d $ in Lemma \ref{lemma:high-lebel-decomposition-debiased-lasso} and  concentration bound for sums of independent random variables to get a high probability bound for the estimation error of $ \ttheta_{0}. $
\end{proof}
Before we study the theoretical properties of the thresholded estimators, we would like to make a few remarks about the rates of convergence for the coordinates of $ \ttheta_0. $ If there is reason to believe that none of the local parameters  are outliers and we consider $ \eta_j $'s large enough, then the global parameter is identified as  mean of  the local parameters, \ie\ $ \thetas_0 = \frac1m \sum_{ k= 1}^m \thetas_k. $ In that case, the rate of convergence for $ \ttheta_0 $ in Lemma \ref{lemma:integrated-hp-bound} simplifies to:
\[  \| \ttheta_{0}- \thetas_0  \|_\infty \lesssim_\Pr \sqrt{\frac{\log d}{m^2}\sum_{ k =1}^m\frac1{n_k}}, \] where, $ \lesssim_\Pr $ implies with probability converging to $ 1. $ This kind of rate is not surprising. Probably the simplest situation, under which this kind of convergence rate arises is ANOVA model. Let $ \left\{Y_{ki}\right\}_{i\in[n_k],k\in[m]} $ is a set of  independent random variables, such that  $ Y_{ki}\sim N(\thetas_0+\deltas_k,\sigma^2), $ where $ \thetas_{0}, \deltas_1,\ldots, \deltas_m $ are some real numbers. Under the assumption $ \sum_{ k =1}^m \deltas_k=0, $ the estimator $ \widehat\theta_0= \frac1m\sum_{ k =1}^m \bar Y_{k\cdot} $, where, $ \bar Y_{k\cdot} = \frac1{n_k}\sum_{ i= 1}^{n_k}Y_{ki}, $ follows distribution $ N\left(\thetas_0, \frac{1}{m^2}\sum_{ k =1}^m\frac1{n_k}\right). $ This gives us the following convergence rate for $ \widehat\theta_0 :$
\[ \Pr\left( |\widehat\theta_0-\thetas_0|\le  3\sqrt{\frac{1}{m^2}\sum_{ k =1}^m\frac1{n_k}} \right) \approx 1.  \]
We notice that $ \frac1{\sum_{ k \in  I_j}n_k} \le \frac{1}{|I_j|^2} \sum_{ k \in  I_j}\frac1{n_k},$ and the equality holds when $ n_k $'s are equal, which gives us the best possible rate. In this case the  rate  simplifies  to:
\[ \left|  (\ttheta_0)_j- (\thetas_0)_j \right| \lesssim_\Pr\sqrt{\frac{\log d}{|I_j|n}} \le \sqrt{\frac{\log d}{(1-\alpha_j)N}}.  
\]

Let $ n_{\text{min}} = \min_k n_k $ and $ \alpha_{\max} = \max_j \alpha_j .$ Notice that $ \frac1{\sum_{ k \in  I_j}n_k} \le \frac{1}{|I_j|^2} \sum_{ k \in  I_j}\frac1{n_k} \le \frac1{|I_j|n_{\min}} \le \frac1{(1-\alpha_j)mn_{\min}}.  $ Hence, we can get the following simple (possibly naive) $ \ell_{\infty} $ high probability  bound for $ \ttheta_{0}: $
\[ \|\ttheta_{0}-\thetas_0 \|_\infty \le 4\sigma \sqrt{\frac{\log d}{(1-\alpha_{\max})mn_{\min}}}.  \]
The estimators in Lemma \ref{lemma:integrated-hp-bound} are dense, they have higher $ \ell_1 $ and $ \ell_2 $ error. We threshold the estimators at a suitable level. The thresholding is usually done at a level of the $ \ell_{\infty} $ error rate of the estimator.
We recall the following result from \cite{lee2017Communicationefficient} which ensures probability convergence for such thresholding. 
\begin{lemma}[\cite{lee2017Communicationefficient}, Lemma 11]
\label{lemma:thresholding-lee}
     As long as $t > \|\bar \beta - \beta^*\|_\infty$ $\bar\beta^{ht}_t = \operatorname{HT}_t(\bar \beta)$ satisfies the following. 
     \begin{enumerate}
         \item $\|\bar \beta^{ht} - \beta^*\|_\infty < 2t,$
         \item $\|\bar \beta^{ht} - \beta^*\|_2 < 2\sqrt{2s}t,$
         \item $\|\bar \beta^{ht} - \beta^*\|_2 < 2\sqrt{2}st,$
     \end{enumerate} where, $s$ is the sparsity of $\beta^*.$ The analogous results holds for $\bar\beta^{st}_t = \operatorname{ST}_t(\bar \beta).$
\end{lemma}
A combination of the Lemmas \ref{lemma:integrated-hp-bound} and \ref{lemma:thresholding-lee} establishes the rate of convergence for the estimates of global and heterogeneous effect parameters. 

\begin{theorem}
Define $ N_{\min} \coloneqq mn_{\min}. $	Assume that the threshold for   $ \ttheta_0 $ is set at $ t_{0} =  4\sigma\sqrt{\frac{\log d}{(1-\alpha_{{\max}})N_{\min}}}$ and for $ \tdelta_k $'s are set at $ t_k = 4\sigma\sqrt{\log d/n_k},  $ respectively.  Then under the conditions in Lemma \ref{lemma:integrated-hp-bound} and for sufficiently large $ n_k $ we have the following:
	\begin{enumerate}
		\item $ \|\widehat\theta_0 - \thetas_0\|_\infty \lesssim_\Pr \sqrt{\frac{\log d}{N_{\min}}}, $ 
		\item $ \|\widehat\theta_0 - \thetas_0\|_1 \lesssim_\Pr s(\thetas_0)\sqrt{\frac{\log d}{N_{\min}}}, $
		\item $ \|\widehat\theta_0 - \thetas_0\|_2 \lesssim_\Pr \sqrt{\frac{s(\thetas_0)\log d}{N_{\min}}}, $  
		\item $ \|\widehat\delta_k-\deltas_k\|_\infty \lesssim_\Pr  \sqrt{\frac{\log d}{n_k}}, $
		\item $ \|\widehat\delta_k-\deltas_k\|_1 \lesssim_\Pr  s(\deltas_k)\sqrt{\frac{\log d}{n_k}}, $
		\item $ \|\widehat\delta_k-\deltas_k\|_2 \lesssim_\Pr  \sqrt{\frac{s(\deltas_k)\log d}{n_k}}, $
	\end{enumerate}
where, for a vector $ \theta\in\reals^d, $ $ s(\theta) $ denotes its sparsity level. 
\label{th:bounds-estimators}
\end{theorem}
It is not necessary to threshold all the co-ordinates of $ \ttheta_{0} $ at a same level. Thresholding the $ j $-th co-ordinate of $ \ttheta_0 $ at  $ t_{0j} = 4\sigma\sqrt{\frac{\log d}{|I_j|^2}\sum_{ k \in  I_j}\frac1{n_k}} $ may give us a faster rate of convergence for  $ \widehat\theta_0 $ in terms of $ \ell_1 $ and $ \ell_2 $ errors. If one wish to use cross-validation to determine an appropriate choices of thresholding, setting different thresholding level for different co-ordinates may be computationally hectic. For computational simplicity, we consider same thresholding level for all the co-ordinates.

One important consequence of Theorem  \ref{th:bounds-estimators} is variable selection under beta-min assumption. For a vector $ \theta\in\reals^d $ let us define its sparsity set as $ \cS(\theta) = \{j \in[d]: \theta_j \neq 0 \}. $ Under the assumption that 
\begin{equation}
    \min_{j\in \cS(\thetas_0)} |\thetas_j| \gg \sqrt{\frac{\log d}{N_{\min}}}  \text{ and, } \min_{j\in \cS(\thetas_k)} |(\thetas_k)_j| \gg \sqrt{\frac{\log d}{n_k}} 
    \label{eq:beta-min}
\end{equation}
  Theorem \ref{th:bounds-estimators} implies $ \Pr\left( \cS(\widehat{\theta}_0) =  \cS (\thetas_0), \cS(\widehat{\theta}_k) =  \cS(\thetas_k) \right) \to 1, $ \ie, \ with very high probability all the active variables will be selected by the proposed estimators.

\section{Computational results}
\label{sec:computationalResults}

In this section, we compare the performance of MrLasso to several other global parameters considered in Section \ref{sec:identifiability} on simulated data:
 the Mean  \eqref{eq:anovaIdRestriction},
    the {Median} \eqref{eq:median}, 
    the square and absolute error trade-off in \eqref{eq:sq-abs} and the Huber estimator in \eqref{eq:huber-id}.
 The datasets are generated from a linear model with $d=2000$ covariates. In our simulation, the covariates are generated from auto-regressive model ($x = (x_1, \dots, x_d)^\top $, $x_1 \sim N(0, 1)$, $x_j = \sqrt{1-\rho^2} x_{j-1} + \rho \eps_j$, and $\eps_j \sim N(0, 1)$) with the correlation $\rho = 0.9$, and the response in the $k$-th local dataset is generated as 
\[
y_k = \bx_k^\top \beta_k^* + \eps_k, ~ \eps_k \sim N(0, 0.25),
\]
where $\beta_k^*$ is the vector of regression coefficients for the $k$-th dataset. The first $s$ co-ordinates of $\beta_k^*$'s are independent $ N(2, (1.5)^2) $ random scalars. The next 20 coordinates of $\beta_k^*$ are drawn from the mixture distribution $(1-\frac1m) \cdot \delta_0 + \frac1m\cdot U([15,16])$, where $\delta_0$ is the degenerate  distribution at zero. This endows the local coordinates with the cluster structure of Assumption \ref{assump:cluster-parameter}. This structure in the local coordinates highlights the difference between the identification restrictions in MrLasso and the other estimators.

\begin{figure}
	\centering
		\includegraphics[width=\textwidth]{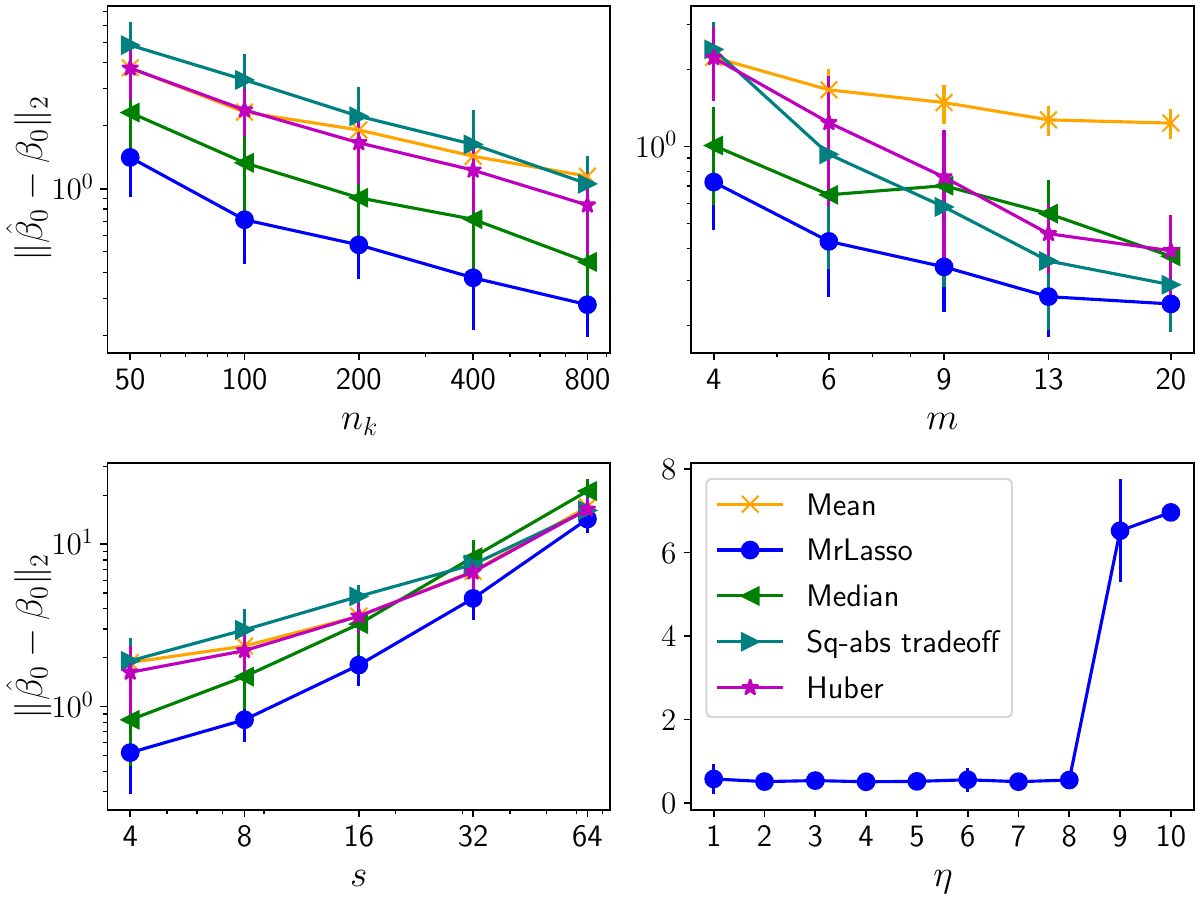}
	\caption{\textit{Upper left}: Errorbar plots for $\ell_2$-error in estimation of $\beta_0$ using different global estimators for varying sample sizes in each datasets ($n_k$) with $m = 5$ and $s = 5.$  \textit{Upper right}: For varying $m$ with $n_k = 200$ and $s = 5$. \textit{Lower left}: For different $s$ with $n_k = 200$ and $m = 5$.   \textit{Lower right}: For different $\eta$ with $n_k = 200$ and $s = 5$ and $m = 5.$}
	\label{fig:l2-m-n-s}
	
\end{figure}

Note that the definition MrLasso, square and absolute loss trade-off and Huber loss require additional tuning parameters.
In the simulation, we hold out $\frac15$ of the dataset as a validation set and we pick  those parameters to minimize test error on the validation set. The data dependent choice of tuning parameters are used to define the global parameters for square and absolute loss trade-off and Huber loss to keep the comparison fair, whereas for MrLasso the global parameters are  identified with  oracle tuning parameter values. 
The upper left plot in Figure \ref{fig:l2-m-n-s} shows the $\ell_2$-estimation error  of the global parameter with respect to the sample size for each dataset ($n_k$). We fix the number of datasets ($m$) at 5 and the sparsity of the MrLasso global parameter ($s$) at $5$. 
Note that the global parameters are different for different identification restrictions. 
The convergence rate of MrLasso is faster than the others. For denser global parameters, the performance discrepancy between the estimators reduces. We see this behavior in the lower-left plot of Figure \ref{fig:l2-m-n-s}.

A study on the effect of $m$ for fixed $n_k$ and $s$ is presented in the upper right plot of Figure \ref{fig:l2-m-n-s}. All the estimates are able to borrow strength across datasets ($\ell_2$-error decreases for higher $m$) but MrLasso produces lower $\ell_2$-errors than others. 

In the lower right panel of Figure \ref{fig:l2-m-n-s} we study the sensitivity of MrLasso to the misspecification of $\eta$. As seen in the plot, MrLasso produces a good estimate of the global parameter as long as $\eta$ falls within a reasonable range. This range depends on the separation between the bulk and outlier local coordinates (see Assumption \ref{assump:cluster-parameter}, condition (B3)). The plot also suggests that  (cross-)validation leads to a value of $\eta$ that falls in this range. We confirm this in Figure \ref{fig:eta-m-n-s}, which shows the chosen values of $\eta$ for the simulations in Figure \ref{fig:l2-m-n-s}.

The $\ell_2$ errors in Figure \ref{fig:l2-m-n-s} are calculated in terms of the corresponding global parameters of the estimators, but the global parameters are different. We complement this experiment by evaluating the prediction error of them on a test dataset that is generated without any outlier local regression coefficients. In particular, the first $s$ coordinates of the local regression coefficient vectors are generated independently from $N(2, (1.5)^2)$, while the rest of the coordinates are zero. In this setup, we expect the MrLasso global parameter to be closer to the population global parameter because it correctly identifies the coordinates of all but the first $s$ coordinates as zero. This is confirmed in Figure \ref{fig:error-m-n-s}.

We finally note that our method is both privacy-preserving (no raw data is sent) and communication efficient (only scalars and vectors are sent). Even the model selection protocol for $\eta$ can be carried out without having to send the raw data: we only need to transfer scalars and vectors. We first send the local debiased estimates to the central machine to calculate $\hat\beta_0$ and $\hat\delta_k$'s for several choices of $\eta$'s. These estimates are sent back to the local machines to calculate prediction error on the holdout datasets. Finally these errors can be transferred back to the global machine to determine the choice of $\eta.$

\begin{figure}
	\centering
		\includegraphics[width=\textwidth]{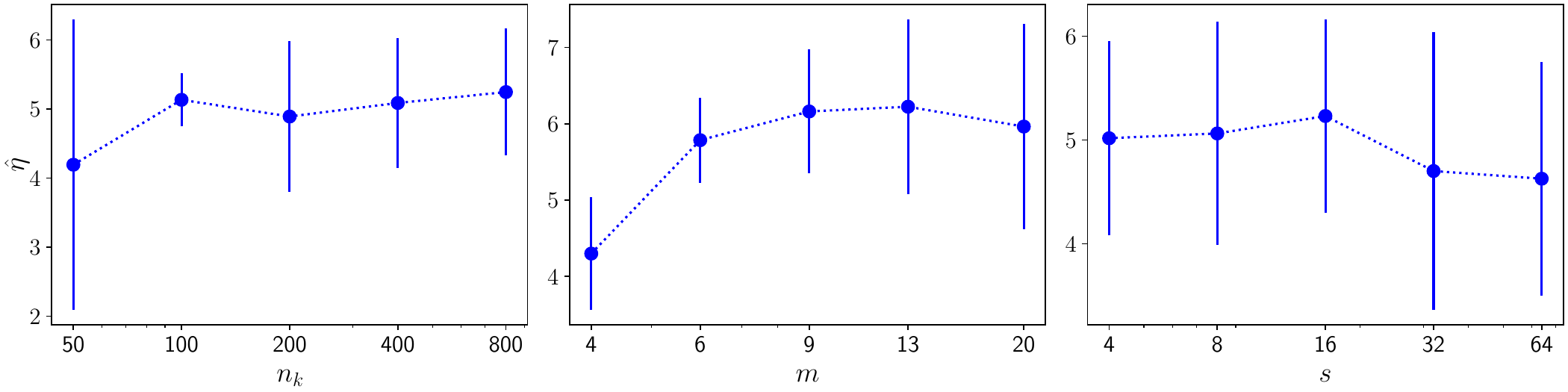}
	\caption{The choices of data dependent $\eta$'s for MrLasso in the experiments in Figure \ref{fig:l2-m-n-s}, upper-left, upper-right and lower-left plot.}
	\label{fig:eta-m-n-s}
	
\end{figure}

\begin{figure}
	\centering
		\includegraphics[width=\textwidth]{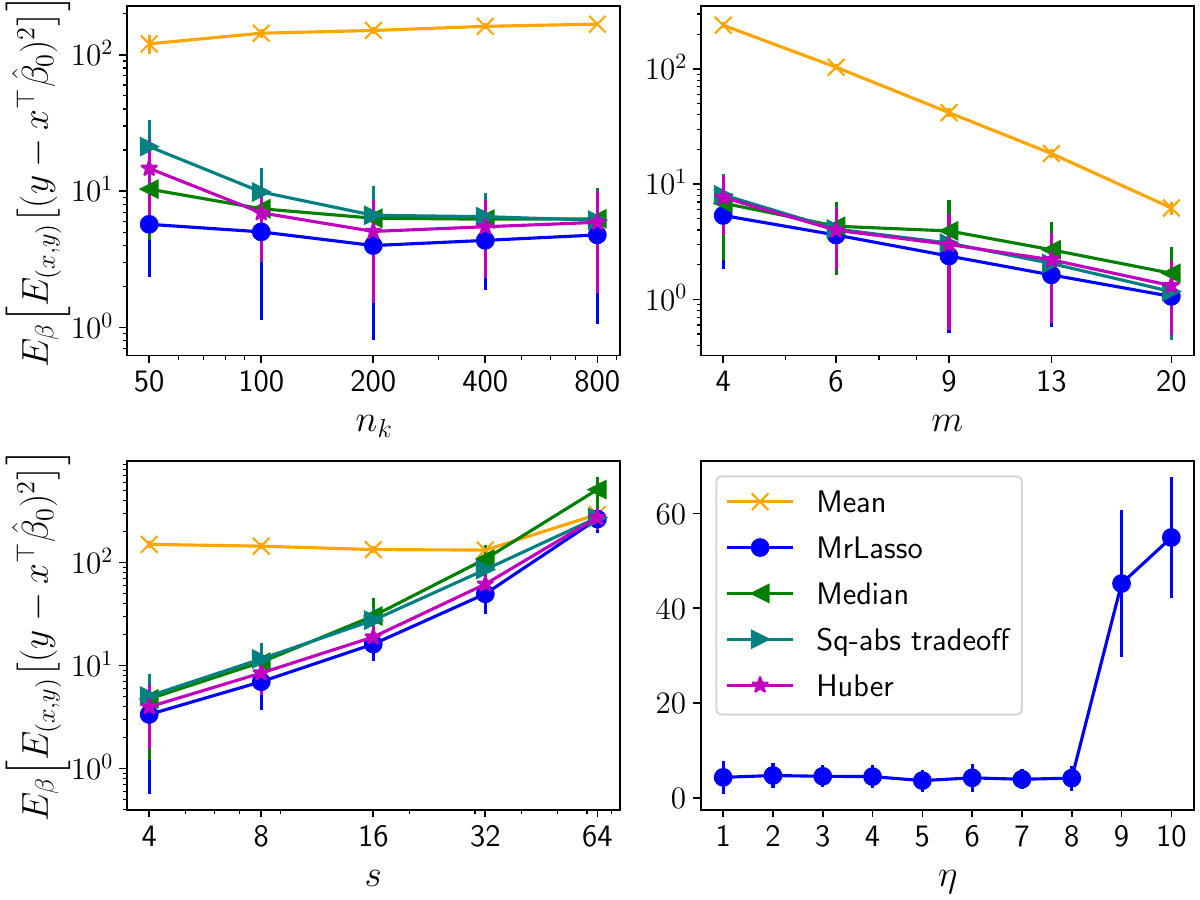}
	\caption{\textit{Upper left}: Errorbar plots for prediction errors in test data  using global estimators for varying sample sizes in each datasets ($n_k$) with $m = 5$ and $s = 5.$  \textit{Upper right}: For varying $m$ with $n_k = 200$ and $s = 5$. \textit{Lower left}: For different $s$ with $n_k = 200$ and $m = 5$.   \textit{Lower right}: For different $\eta$ with $n_k = 200$ and $s = 5$ and $m = 5.$ }
	\label{fig:error-m-n-s}
	
\end{figure}

\section{Cancer cell line study}
The Cancer Cell Line Encyclopedia is a database of gene expression, genotype, and drug sensitivity data for human cancer cell lines. We use our method to study the sensitivity of cancer cell lines to certain anti-cancer drugs. We treat the area above the dose-response curve as the response variable \citep{barretina2012cancer} and use expression levels of approximately $20000 $ genes ($d$) as features. More details about data source and pre-processing are provided in Appendix \ref{sec:supp-ccle}. After prepossessing we have 482 cancer cell lines. Each of these cell lines comes from a cancerous organ and corresponds to a specific cancer type. The lines are divided into different machines ($k$) according to cancer types.   For a given drug, our goal is to obtain a single regression coefficient vector which can be used to predict the response of different  cancer types to that  drug. We fit our model on the more common cancer cell types \footnote{The number of local datasets ($m$) that are used in meta-analysis is same as the number of common cancer cell types shown in the upper left panel of the corresponding figures (for example the drug PD-0325901 has $m = 7$ as seen in Figure \ref{fig:ccle-PD-0325901}), where the number of  data-points ($n_k$'s) in each of these common cancer types are also shown in the upper-left panel.} with at least 10 samples and evaluate the prediction accuracy of the models on rare cancer types (those with less than 10 samples).  The performance of three models are presented: soft thresholded mean of the local estimators (which we call Mean), the soft thresholded version of our integrated estimator (as in section \ref{sec:estimator}), and a global lasso estimator that simply fits the lasso to an aggregate dataset consisting of samples from all common cancer types.

\begin{figure}
 \centering
 \includegraphics[width=\textwidth]{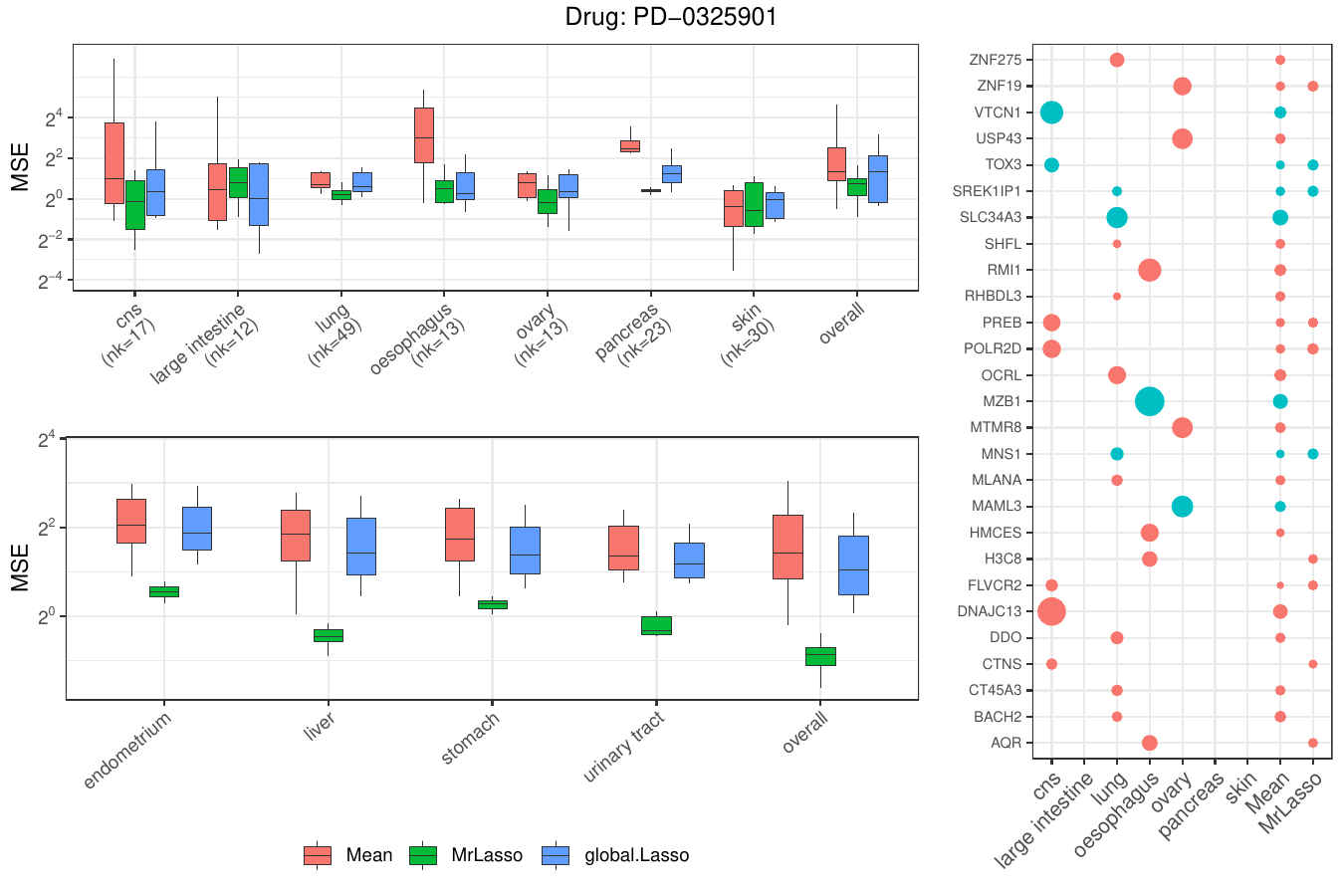}

 \caption{\textit{Left:} Drug-response prediction mean squared errors (MSEs) for  the most frequent cancer types (upper left) and rare cancer types (lower left panel) for the drug PD-0325901.  \textit{Right:} Coefficient plot for the selected genes in lasso estimates for most frequent cancer types (calculated locally on each machine) and Mean and Mr Lasso (integrated across machines). The size of the points is related to the magnitude of the coefficients, whereas the signs are represented by colour (red for negative and blue for positive).}
 \label{fig:ccle-PD-0325901}
\end{figure}

\begin{figure}
 \centering
 \includegraphics[width=\textwidth]{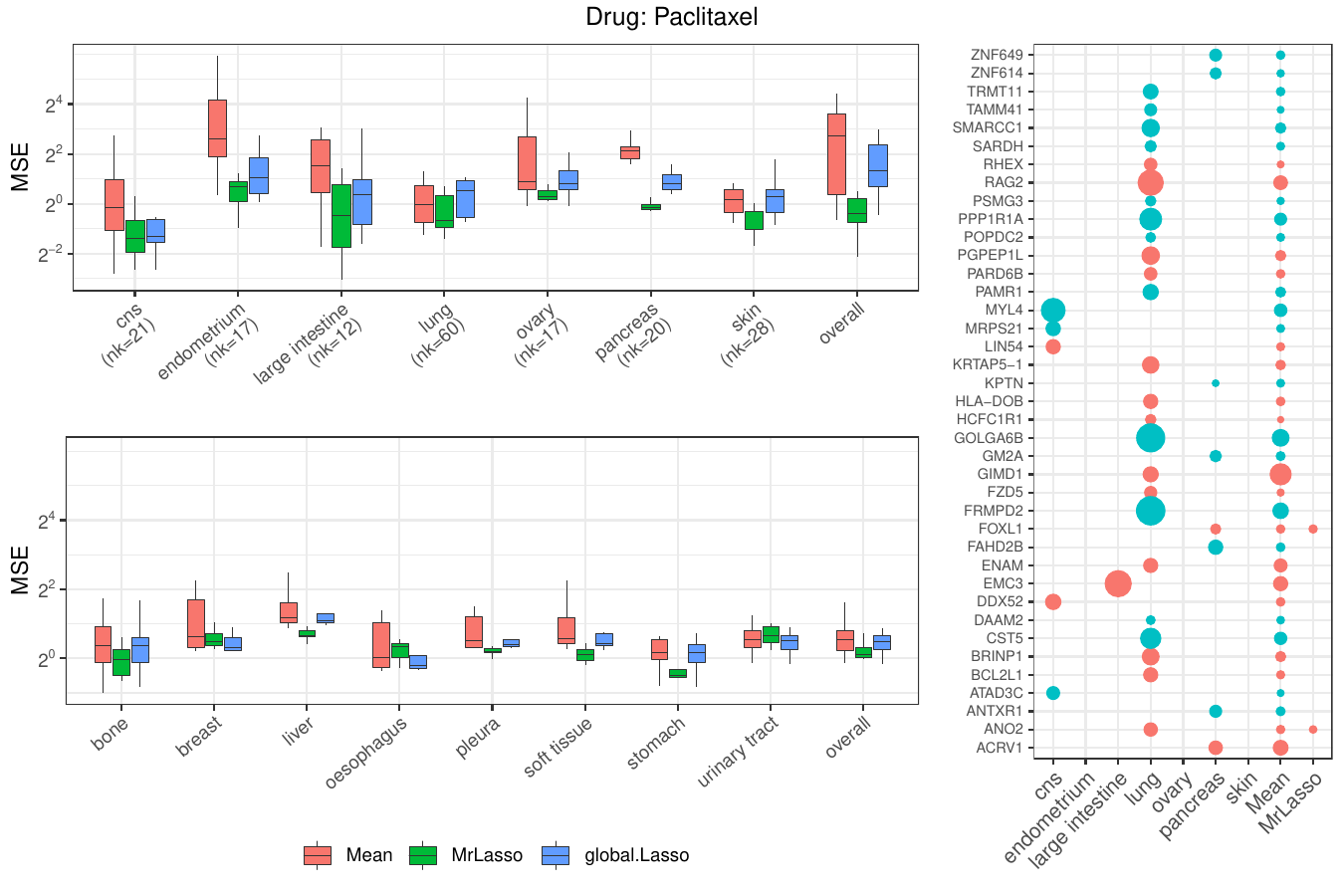}

 \caption{\textit{Left:} Drug-response prediction mean squared errors (MSEs) for  the most frequent cancer types (upper left) and rare cancer types (lower left panel) for the drug Paclitaxel.  \textit{Right:} Coefficient plot for the selected genes in lasso estimates for most frequent cancer types and Mean and Mr Lasso.}
 \label{fig:ccle-Paclitaxel}
\end{figure}

\begin{figure}
 \centering
 \includegraphics[width=\textwidth]{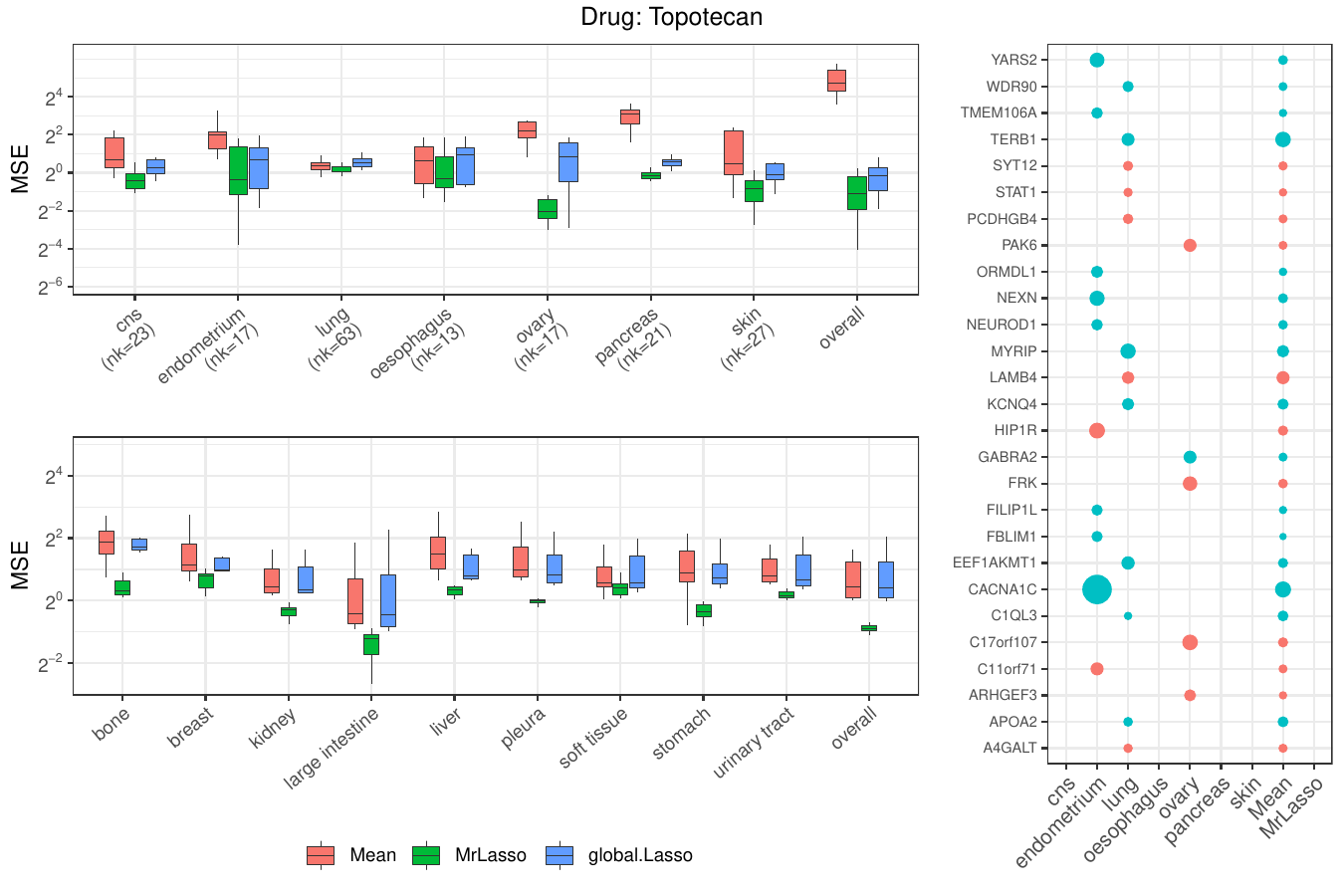}

 \caption{\textit{Left:} Drug-response prediction mean squared errors (MSEs) for  the most frequent cancer types (upper left) and rare cancer types (lower left panel) for the drug Topotecan.  \textit{Right:} Coefficient plot for the selected genes in lasso estimates for most frequent cancer types and Mean and Mr Lasso.}
 \label{fig:ccle-Topotecan}
\end{figure}

For integration purposes we first need to calculate the de-biased lasso estimate from each of these frequent cancer types. The regularization parameter for each of these cancer types is chosen using 7 fold cross-validation. After calculating the debiased lasso from each dataset, Mean is obtained by soft-thresholding their average, with the level of thresholding determined via  10 fold cross-validation. To calculate  MrLasso, we assume that the parameters $\eta_j$ are identical over the gene expressions. Besides determining this common parameter $\eta$ we  also require a proper soft-threshold level ($t$) for the integrated estimator. To this end the appropriate pair $(\eta,t)$ is again chosen  via cross validation (see the discussion in last paragraph of Section \ref{sec:estimator}).  
We compare the performances of the Mean, the MrLasso and the global lasso (obtained using combined data over the most-frequent cancer types with  best regularization parameter chosen via 10-fold cross-validation) for drugs PD-0325901 (Figure \ref{fig:ccle-PD-0325901}), Paclitaxel (Figure \ref{fig:ccle-Paclitaxel}) and Topotecan (Figure \ref{fig:ccle-Topotecan}).

To evaluate the efficacy of this approach, we fit the global parameter using Mean, MrLasso, and a global lasso on the more common cancer types and evaluate its predictive accuracy on rare cancer types. The lower left panels in Figures \ref{fig:ccle-PD-0325901}, \ref{fig:ccle-Paclitaxel} and \ref{fig:ccle-Topotecan} give comparative plots for the  prediction accuracy of drug-response of the three global parameter estimates (for the three drugs). The plots for several other drugs are provided in Appendix \ref{sec:supp-ccle}. For completeness, we also evaluate the predictive performance of the fitted global parameters on held out data from the more common cancer types in the top left panels.

We see that the predictive performance of MrLasso is generally superior to that of Mean and the global lasso. This is mostly due to the heterogeneity in the relationship between drug sensitivity and genotype on different cancer types, as evinced by the  considerable variation between the effect sizes of genetic variants on drug sensitivity from cancer to cancer: see for example,  Figures \ref{fig:ccle-PD-0325901}, \ref{fig:ccle-Paclitaxel} and \ref{fig:ccle-Topotecan}. This causes the global parameter fitted by the Mean to be considerably denser than the global parameter fitted by MrLasso. In other words, the Mean pools effect sizes across all cancer types regardless of whether the effect sizes are similar. This causes the global parameter fitted by the Mean to pick up many small effect sizes that are specific to certain cancer types, thereby detracting from the generalizability of the global parameter to other (rare) cancer types. On the other hand, MrLasso only aggregates effect sizes when the data suggests they are similar. In other words, the global parameter fitted by MrLasso only incorporates effects that \emph{persist} across cancer types. This leads to a parameter whose prediction performance generalizes across cancer types. In some extreme cases (\eg\ Figure \ref{fig:ccle-Topotecan}), the effect sizes are so dissimilar that MrLasso does not borrow strength at all, but the Mean continues to pool the effect sizes. This leads to overall poor prediction performance.
In comparison with Mean, we see that  MrLasso can  \emph{automatically adapt} to the intrinsic level of heterogeneity across datasets, leading to a better generalization.

\section{Discussion}
\label{sec:discussion}

We consider integrative regression in the high-dimensional setting with heterogeneous data sources. The two main issues that we address are (i) identifiability of the global parameter, and (ii) statistically and computationally efficient estimator of the global parameter. 
In many prior works on integrative regression in high dimensions, there is either no global parameter or a global parameter that is not properly identified (see Section \ref{sec:identifiability} for a discussion of global parameters in prior work). 

We suggest a way to identify the global parameter that addresses some of the drawback of prior approaches by appealing to ideas from robust location estimation. The main benefit of our suggestion is it is possible to estimate the global parameter at a rate that depends on the size of the combined dataset. We also proposed a statistically and computationally efficient estimator of the global parameter. By statistically efficient, we mean the estimation error vanishes at a rate that depends on the size of the combined datasets. By computationally efficient, we mean the communication cost of evaluating the estimator depends only on the dimension of the problem and the number of machines. Further, because no individual samples are communicated between machines, evaluating the proposed estimator does not compromise the privacy of the individuals in the data sources.

\bibliography{YK,sm}
\newpage
\appendix
\section{Identifiability}
\label{sec:counter-example}

\begin{lemma}[Contaminated normal model]
\label{lemma:contaminated-normal}
Let \[ x  \sim (1-\eps) \cN\big((\mu_0, \sigma^2\big) + \eps g,\] where $\eps<1/2$ and $g$ is continuous density of contamination distribution. The expected redescending loss $L(\mu) = E[(x - \mu)^2 \wedge \eta^2]$ is minimized at $\mu_0$ if 
\[
\int_{\mu_0 - \eta}^{\mu_0 + \eta} (x - \mu_0) g(x)dx = 0
\]
and 
\[ (1-\eps) \Bigg\{ \Phi\Big(\frac\eta\sigma\Big) - \Phi\Big(-\frac\eta\sigma\Big) - \frac{2\eta}\sigma \phi\Big(\frac\eta\sigma\Big)  \Bigg\} + \eps\Bigg\{ G(\mu + \eta) - G(\mu-\eta) - 2\eta \big(g(\mu + \eta) + g(\mu - \eta)\big)  \Bigg\} > 0, \] where $\phi$ and $\Phi$ are density and distribution functions of standard normal distribution. 
\end{lemma}

\begin{proof}
    We note that 
    \[
    \begin{aligned}
     L'(\mu)/2& = E\big[ (\mu - x) \mathbbm{1}(|x - \mu| \le \eta) \big]\\
     & = (1 - \eps )E_{(1/\sigma)\phi\big((x - \mu)/\sigma\big)}\big[ (\mu - x) \mathbbm{1}(|x - \mu| \le \eta) \big]\\
     & \quad+ \eps E_g\big[ (\mu - x) \mathbbm{1}(|x - \mu| \le \eta) \big]
    \end{aligned}
    \] 
    Here \[
    \begin{aligned}
    &  E_g\big[ (\mu_0 + \delta - x) \mathbbm{1}(|x - \mu_0 - \delta| \le \eta) \big]\\
    &\quad = \delta \int_{\mu_0 + \delta - \eta} ^ {\mu_0 + \delta + \eta} g(x)dx - \int_{\mu_0 + \delta - \eta} ^ {\mu_0 + \delta + \eta} (x - \mu_0)g(x)dx\\
     & \quad = \delta \Big[G(\mu_0 + \delta + \eta) - G(\mu_0 + \delta - \eta)\Big] - \Big[\Gamma(\mu_0 + \delta + \eta) - \Gamma(\mu_0 + \delta - \eta)\Big],
    \end{aligned}
    \] where $\Gamma'(x) = (x-\mu_0) g(x).$ From the first order Taylor expansion 
    \[
    \begin{aligned}
    &  E_g\big[ (\mu_0 + \delta - x) \mathbbm{1}(|x - \mu_0 - \delta| \le \eta) \big]\\
    & \quad \approx \delta \Big[G(\mu_0 + \eta) + \delta g(\mu_0 + \eta) - G(\mu_0  - \eta) - \delta g(\mu_0 - \eta) \Big]\\
    & \quad \quad - \Big[ \Gamma(\mu_0 + \eta) + \delta \eta g(\mu_0 + \eta ) - \Gamma(\mu_0 - \eta) + \delta \eta g(\mu_0 - \eta) \Big]\\
    & \quad \approx \delta \Big[ G(\mu_0 + \eta) - G(\mu_0 - \eta ) - \eta \big( g(\mu_0 + \eta) + g(\mu_0 - \eta) \big) \Big] -  \Big[ \Gamma(\mu_0 + \eta)  - \Gamma(\mu_0 - \eta)\Big]
    \end{aligned}
    \]

    If $g$ is the density of $N(\mu_0, \sigma^2)$ then 
   \[ \begin{aligned}
    &  E_g\big[ (\mu_0 + \delta - x) \mathbbm{1}(|x - \mu_0 - \delta| \le \eta) \big]\\
    & \quad \approx \delta \Bigg[\Phi\Big(\frac\eta\sigma\Big) - \Phi\Big(-\frac\eta\sigma\Big) - \frac{2\eta}\sigma \phi\Big(\frac\eta\sigma\Big) \Bigg].
    \end{aligned}
    \]
    
From the requirement that 
\[
L'(\mu_0+ \delta) = \begin{cases}
0 & \text{if}\ \delta = 0,\\
> 0 & \text{if}\ \delta > 0, \\
<0 & \text{if} \ \delta < 0 ,
\end{cases}
\] for any arbitrarily small $\delta$, we have the lemma. 
    
\end{proof}

\begin{example}
\label{ex:medianFail}
Consider an one dimensional case where we have $ 2m+1 $ datasets each of size $ n. $ The  local parameters of these datasets take the values $ \mu_k^* = 2k $ for $ k=1,2,\cdots,2m+1. $ In this case the global parameter $ \mu_0^* $ will be identified as $ 2(m+1). $ Suppose we have some estimators for the local parameters $ \widehat\mu_k $ which can be written as $ \widehat\mu_{k}= \mu_k^* + \eta_k $ and for simplicity we assume $ \eta_k $ are  normal random variables with mean zero standard deviation  $  \frac\sigma{\sqrt{n}}.$  The example can also be extended for sub-Gaussian random variables. 

 The parameter $ \sigma  $ depends on the variances of observations. For simplicity we assume that $ \sigma  $ is $ 1. $ For some $ 0<\delta <1 $ if  sample sizes $ n\ge  8\log \left( \frac{4m+2}\delta \right) $ for each datasets then by union bound over sub-Gaussian inequality we get 

\begin{align*}
	\Pr \left(  \left| \widehat\mu_{k}-\mu_k^* \right| \le \frac12 \  \text{for all } k=1,2,\cdots,2m+1 \right) & \ge 1- \sum_{ k =1}^{2m+1} \Pr \left( |\eta_k|>t  \right)\\& \ge 1- 2(2m+1) e^{-\frac{n}{8}}\\
	& \ge 1-\delta.
\end{align*}
 We can see that with probability at least $ 1-\delta $ the local estimators are  ordered as $ \widehat\mu_1< \widehat\mu_2 < \cdots <\widehat\mu_{2m+1}  $. Let this event be $ E. $  A natural choice for $ \widehat \mu_0 $ is the median of $ \widehat\mu_{k} $'s under  the identification condition   \eqref{eq:median}. Clearly, under the event $ E, $ $ \widehat\mu_0 $ is $ \widehat\mu_{m+1} $ and hence we have $ \widehat\mu_0 - \mu_0^* = \widehat\mu_{m+1} - \mu_{m+1}^* = \eta_{m+1}.$ As $ \eta_{m+1}  $ is  $ N(0,1/n) $, we see that
  \begin{align*}
  	\Pr\left( |\widehat\mu_0-\mu_0^*| > \frac C{\sqrt{mn}}\right) & \ge \Pr\left(|\eta_{m+1}|>\frac C{\sqrt{mn}}\right) - \Pr(E^c)\\
  	& \ge  2 - 2\Phi \left( \frac C{\sqrt{m}}  \right) - \delta .
  \end{align*}
  For  a fixed  $ C>0  $ the  above probability is larger than $ 1- \Phi \left( \frac C{\sqrt{m}}  \right) $ whenever $ n\ge 8\log \left( \frac{4m+2}{1 - \Phi \left( \frac C{\sqrt{m}}  \right)} \right) . $ This is a contradiction to the fact that  rate of convergence for $ \widehat\mu_0 $ to $ \mu_0^* $ is $ \frac{1}{\sqrt{mn}}. $

\end{example}

\begin{example}

 Let $\mu_1 = \dots = \mu_{m-1} = 0,$ and $\mu_m >0.$ For any choice of $\lambda>0,$ $\mu_0$ as defined in \eqref{eq:huber-id} is never zero. 
    \label{ex:huber-loss-counter}
\end{example}

\begin{proof}
Define \[L(x,\lambda) = \sum_{k=1}^m L_\lambda(x-\mu_k) = (m-1) L_\lambda (x) + L(\mu_m - x).\]

Since \[
L_\lambda (x)  = \begin{cases} \frac12 x^2 & \text{if } |x|\le \lambda \\
\lambda \left(|x| - \frac12 \lambda \right) & \text{otherwise},
\end{cases}
\] we note that 
\[
\partial_x L_\lambda (x)  = \begin{cases}  x & \text{if } |x|\le \lambda \\
\lambda \cdot \text{sign}(x)  & \text{otherwise},
\end{cases}
\] for $x \neq 0$ and $\partial_x L_\lambda(0) = 0$, we have 
\begin{equation}
 \partial_xL(0, \lambda) = (m-1) \partial_x L_\lambda (0) + \partial_x L(\mu_m) = \lambda\cdot \text{sign}(\mu_m)\,.
 \label{eq:huber-derivative}
\end{equation} Note that $\partial_xL(0, \lambda)$ is can never be zero, which concludes that the minimum is never achieved at zero.

\end{proof}

\begin{result}
	Suppose the covariate dimension $ d $ is fixed. Assume
	
	\begin{enumerate}
		\item $ n_1 = \cdots=n_m = n, $
		\item $ \text{Var}(\beps_{ki}) $ are same over different $ k $'s,
		\item $ \Sigma_k = \Ex \left(\bx_{ki}\bx_{ki}^T\right)  $'s  are invertible, 
		\item $ \sum_{ k =1}^m\|\thetas_k -\theta_0\|_1 + c\|\theta_0\|_1 $ has a unique minimizer $ \thetas_0. $
	\end{enumerate}
	   Define \[ (\htheta_0, \hdelta_1, \ldots, \hdelta_m)(n) = \argmin_{(\theta_0,\delta_1, \dots , \delta_m)} \frac{1}{2N} \sum_{k=1}^m  \left\| \bY_{k} - \bX_{k} (\theta_0+ \delta_k)  \right\|_2^2 + \lambda_n \sum_{k=1}^m \|\delta_k\|_1 + c\lambda_n\|\theta_0\|_1, \] for $ \lambda_n \sim \frac1{\sqrt{n}} $ some $ c>0 $  independent of $ n. $ Then \[ (\htheta_0, \hdelta_1, \ldots, \hdelta_m)(n) \to (\thetas_0, \deltas_1, \ldots, \deltas_m ), \] in probability, with respect to $\|\cdot\|_2$ norm, where, $ \thetas_0 = \argmin_{\theta_0} \left(\sum_{ k =1}^m \|\thetas_k -  \theta_0 \|_1 + c\|\theta_0\|_1  \right) $ and  $ \deltas_k = \thetas_k - \thetas_0 \text{ for } k \in [m]. $
	   \label{result:id-data-shared-lasso}
\end{result}
This result can be shown for any norm in $\reals^{d(m+1)}.$
\begin{proof}
	Define   \[ M_n(\phi) =  \frac{1}{2N} \sum_{k=1}^m  \left\| \bY_{k} - \bX_{k} (\theta_0+ \delta_k)  \right\|_2^2 + \lambda_n \sum_{k=1}^m \|\delta_k\|_1 + c\lambda_n\|\theta_0\|_1 \]  and \[ \widetilde M_n(\phi) =  \frac{1}{2m} \sum_{k=1}^m  (\theta_0 + \delta_k -\thetas_k)^T\Sigma_k (\theta_0+\delta_k - \thetas_k) + \lambda_n \sum_{k=1}^m \|\delta_k\|_1 + c\lambda_n\|\theta_0\|_1 ,\] 
	where, $ \phi = (\theta_0, \delta_1, \ldots, \delta_m). $ We define $ K_M  \subset \reals^{d(m+1)} $ to be the closed ball around $ \phi ^* = (\thetas_0,\thetas_1-\thetas_0, \dots, \thetas_m - \thetas_0)$ with radius $ M. $  Since, for each $ \phi\in K_M, \ M_n(\phi) \to \widetilde{M}_n(\phi) $ and $ \sup_{\phi \in K_M} \widetilde{M}_n(\phi)  $ is finite, $ \sup_{\phi \in K_M} |M_n(\phi) - \widetilde{M}_n(\phi)| \to 0 $ in probability. 
	
	For any $ \epsilon>0,  $ \[ \inf_{\|\phi - \phi^*\|_2\ge \epsilon}\widetilde M_n(\phi) > \widetilde{M}_n (\phi^*).  \] So, by argmin continuous mapping theorem \[ \left \|\argmin_{\phi \in K_M} M_n(\phi) - \argmin_{\phi\in K_M} \widetilde{M}_n(\phi)\right \|_2 \to 0  \] in probability. Taking the radius $ M \to \infty $ we get \[\left \| \argmin_{\phi \in \reals^{d (m+1)}} M_n(\phi) - \argmin_{\phi\in \reals^{d (m+1)}} \widetilde{M}_n(\phi)\right \|_2 \to 0 \]
	
	As $ n \to \infty , \ \lambda_n \to 0+. $ Hence, in the minimization problem of $ \widetilde{M}_n(\phi) $ $ \frac{1}{2m} \sum_{k=1}^m  (\theta_0 + \delta_k -\thetas_k)^T\Sigma_k (\theta_0+\delta_k - \thetas_k) $ becomes more and more important, and in the limiting case this becomes primary objective. Looking at the loss for primary objective we see that optimum achieved at $ \phi  $ for which $ \theta_0+\delta_k = \thetas_k $ for all $ k\in[m]. $ Hence, $ \argmin_{\phi\in \reals^{d\times (m+1)}} \widetilde{M}_n(\phi) \to \argmin_{\phi \in A}  \sum_{k=1}^m \|\delta_k\|_1 + c\|\theta_0\|_1, $ where, $ A = \{\phi: \theta_0 + \delta_k = \thetas_k \}. $

\end{proof}

\section{Supplementary Results}
\label{subsec:supp-results}

\begin{lemma}
\label{lemma:high-lebel-decomposition-debiased-lasso}
Under the Assumption \ref{assump:debiased-lasso-assumption} the followings hold. 
    \begin{enumerate}
		\item\label{assump:lin-decompose} The debiased lasso estimator $ \ttheta^d $  can be decomposed as 
		
	\[ 	 \ttheta^d = \theta^* - \Thetah \nabla \ell (\theta^*) + \Delta  \]

		where, for  and  some $ c_1>0$, \[\Pr \left(\|\Delta\|_\infty  >c_1   Ks_0\log d/n\right) \le o(d^{-1}).  \]
		\item\label{assump:row-theta} For some $ c_2>0, $  \[\Pr \left(\text{for some }j\in[d], \   \|\Thetah_{j,\cdot}-\Theta_{j,\cdot}\|_1 >c_2  \max \{K s_j, K^2 s_0\} \sqrt{\log d /n}\right) \le o(d^{-1}). \]
		\item There is some $ c_3>0$ such that $ \  \|\Theta_{j,\cdot}\|_1 \le c_3\sqrt{s_j} $ hold for any $ j\in[d]. $
    \end{enumerate}
\end{lemma}

\begin{sproof}{Proof of Lemma \ref{lemma:high-lebel-decomposition-debiased-lasso}}
   For the proof of  2. readers are suggested to see Theorem 3.2. in \cite{vandegeer2014asymptotically}.
	
	Denote $s^* = \max\{s_1, \dots, s_d\}$. 
	To verify 1. we start with the Taylor expansion of $ \ttheta^d $:
	\begin{align*}
		\ttheta^d &= \ttheta - \Thetah \nabla \ell(\ttheta) \\
		& = \ttheta - \Thetah \frac1n\sum_{ i= 1}^n \drho(\by_i,\bx_i^T \ttheta)\\
		& = \ttheta - \Thetah \frac1n\sum_{ i= 1}^n \drho(\by_i,\bx_i^T \thetas)\bx_i - \Thetah \frac1n\sum_{ i= 1}^n \ddrho(\by_i,\ta_i)\bx_i\bx_i^T (\ttheta-\thetas)\\
		& = \thetas - \Thetah\nabla \ell(\thetas) +  (I - \Thetah \bM)(\ttheta - \thetas) 
	\end{align*}
	where, $ \ta_i $ is some number between $  \bx_i^T\ttheta $ and $ \bx_i^T\thetas, $ and $  \bM =\frac1n\sum_{ i= 1}^n \ddrho(\by_i,\ta_i)\bx_i\bx_i^T.  $ We give a high probability $ \ell_\infty  $ bound  for $  \Delta = (I - \Thetah \bM)(\ttheta - \thetas).  $ By triangle inequality 
	\begin{align*}
		\|(I - \Thetah \bM)(\ttheta - \thetas)\|_\infty \le&  \left\| \left(I - \Thetah \nabla \ell(\ttheta) \right)(\ttheta - \thetas) \right\|_\infty + \left \| \Thetah \left( \nabla^2 \ell(\ttheta)- \bM \right)(\ttheta - \thetas) \right \|_\infty  \\
		\le &  \underbrace{\max_j \left\| e_j^T - \Thetah_{j,\cdot} \nabla^2 \ell(\ttheta)\right \|_\infty \|\ttheta -  \thetas\|_1}_{I}\\&  + \underbrace{\frac1n \sum_{ i= 1}^n \|\Thetah\bx_i\|_\infty |(\ddrho(\by_i, \bx_i^T\ttheta)- \ddrho(\by_i,\ta_i))\bx_i^T(\ttheta-\thetas)|}_{II} 
	\end{align*}
	
	From KKT condition for nodewise lasso we get 
	\[ -\frac1n \bX_{\ttheta, -j} (\bX_{\ttheta, j}- \bX_{\ttheta, -j}\widehat\gamma_j) + \lambda_j \widehat z_j = 0,  \] where, $ \|\widehat z_j\|_\infty \le 1. $ This implies  \[ -\frac1n \widehat\gamma_j^T\bX_{\ttheta, -j} \bX_{\ttheta}\Thetah_{j,\cdot}^T \widehat\tau_j^2 + \lambda_j \|\widehat\gamma_j\|_1 = 0 \] Now 
	\begin{align*}
		\widehat\tau_j^2 &= \frac1n \|\bX_{\ttheta, j}- \bX_{\ttheta, -j}\widehat{\gamma}_j\|_2^2 + \lambda_j\|\widehat{\gamma}_j\|_1\\
		& = \frac1n\bX_{\ttheta, j}^T \bX_{\ttheta}\Thetah_{j,\cdot}^T \widehat\tau_j^2  -\frac1n \widehat\gamma_j^T\bX_{\ttheta, -j} \bX_{\ttheta}\Thetah_{j,\cdot}^T \widehat\tau_j^2 + \lambda_j \|\widehat\gamma_j\|_1\\
		& =  \frac1n\bX_{\ttheta, j}^T \bX_{\ttheta}\Thetah_{j,\cdot}^T \widehat\tau_j^2.
	\end{align*}
	 Hence $ \frac1n\bX_{\ttheta, j}^T \bX_{\ttheta} \hTheta_{j,\cdot}^T = 1. $ This implies  \[ \|e_j - \nabla^2 \ell(\ttheta) \Thetah_{j,\cdot}^T\|_\infty= \left\| \frac1{n\widehat\tau_j^2} \bX_{\ttheta, -j} (\bX_{\ttheta, j}- \bX_{\ttheta, -j}\widehat\gamma_j)\right \|_\infty \le \frac{\lambda_j}{\widehat\tau_j^2} \lesssim \frac1{\widehat\tau_j^2} \left( \frac{\log d}{n}\right)^{\frac12}.  \]
	 By \cite{vandegeer2014asymptotically}, Theorem 3.2, 
	 \[ |\widehat\tau_j^2 - \tau_j^2 |  \lesssim \left( \frac{\max \{K s^*,K^2s_0\}\log d}{n}\right)^{\frac12} \] with probability at least $ 1-o(d^{-1}). $ Thus $ \max_j \left\| e_j^T - \Thetah_{j,\cdot} \nabla^2 \ell(\ttheta)\right \|_\infty \lesssim_\Pr  \left( \frac{\log d}{n}\right)^{\frac12}  $ and by 6. in the Assumption \ref{assump:debiased-lasso-assumption} \[ \max_j \left\| e_j^T - \Thetah_{j,\cdot} \nabla^2 \ell(\ttheta)\right \|_\infty \|\ttheta -  \thetas\|_1 \lesssim_\Pr \frac{s_0\log d}{n}, \] where, $ \lesssim_\Pr $ denotes $ \lesssim $ with probability at least $ 1-o(d^{-1}). $
	 
	 We turn our attention to $ (II). $ We see that 
	 \begin{align*}
	 	\|\Thetah \bX^T\|_\infty \le &  \max_j \|\Thetah_{j,\cdot} \bX^T\|_\infty \lesssim \max_j \|\Thetah_{j,\cdot} \bX_{\thetas}^T\|_\infty\\
	 \le	&  \max_j \frac1{\widehat\tau_j^2} \|\bX_{\thetas, j}-\bX_{\thetas, -j}\widehat\gamma_j\|_\infty.
	 \end{align*}
Again by \cite{vandegeer2014asymptotically}, Theorem 3.2,  
	\begin{align*}
	\lesssim_\Pr & \max_j \frac1{\tau_j^2} \|\bX_{\thetas, j}-\bX_{\thetas, -j}\widehat\gamma_j\|_\infty\\
		\le & \max_j \frac1{\tau_j^2} \|\bX_{\thetas, j}-\bX_{\thetas, -j}\gamma_j\|_\infty\\
		& + \frac1{\tau_j^2}\|\bX_{\thetas, -j}\|_\infty \|\widehat\gamma_j-\gamma_j\|_1,
	\end{align*}
	which, by 1. in the Assumption \ref{assump:debiased-lasso-assumption} and \cite{vandegeer2014asymptotically}, Theorem 3.2, 
	  \[ \lesssim_\Pr K+ K \times \max\{Ks^*, K^2 s_0\}\sqrt{\frac{\log d}{n}}. \]
	 
	  Now, \begin{align*}
	  	& \frac 1n \sum_{ i= 1}^n \|\Thetah\bx_i\|_\infty |(\ddrho(\by_i, \bx_i^T\ttheta)- \ddrho(\by_i,\ta_i))\bx_i^T(\ttheta-\thetas)|\\
	  	& \lesssim_\Pr K \frac1n \sum_{ i= 1}^n  |(\ddrho(\by_i, \bx_i^T\ttheta)- \ddrho(\by_i,\ta_i))\bx_i^T(\ttheta-\thetas)|
	  \end{align*}
	  which, by 5. and 6. in the Assumption \ref{assump:debiased-lasso-assumption}, is at most \[ \lesssim\frac Kn \|\bX(\ttheta- \thetas)\|_2^2 \lesssim_\Pr \frac{K s_0\log d}{n}.  \] By union bound \[ \|\Delta\|_\infty \lesssim_\Pr \frac{(1 + K)s_0\log d}{n} = \cO \big(Ks_0 \log d/n\big). \] This shows assumption (A3).
	  
	  From 3. in the Assumption \ref{assump:debiased-lasso-assumption} that the minimum eigen-value of $ \Sigma_\thetas $ is bounded away from zero for all $ d, $ we get that, for some $ \kappa>0, $ which doesn't depend on $ d, $ the  largest eigne-value of $ \Theta^2 \le \kappa. $  Hence, \begin{align*}
	  	\|\Theta_{j,\cdot}\|_1^2 & \le s_j \|\Theta_{j,\cdot}\|_1^2 \\
	  	& =  s_j e_j^T \Theta^2 e_j, \ \text{where, } \{e_j\}_{j= 1}^d \text{ is the standard basis of }\reals^d\\
	  	& \le s_j \kappa.
	  \end{align*}
	  This shows 3. in lemma \ref{lemma:high-lebel-decomposition-debiased-lasso}
\end{sproof}

\begin{sproof}{Proof of Lemma \ref{lemma:local-estimate-tail-bound}}
We start with  the  linear decomposition,  1. in lemma \ref{lemma:high-lebel-decomposition-debiased-lasso}, which give us	
\[ \|\ttheta^d - \theta^*\|_\infty \le  \| \Thetah \nabla \ell (\theta^*)\|_\infty + \|\Delta\|_\infty. \] 

\begin{align*}
	\|\ttheta^d - \theta^*\|_\infty & \le  \| \Thetah \nabla \ell (\theta^*)\|_\infty + \|\Delta\|_\infty\\
	&\le \underbrace{\| \Theta\nabla \ell (\theta^*) \|_\infty}_{I} +  \underbrace{\|(\Thetah-\Theta)\nabla \ell (\theta^*)\|_\infty}_{II}  + \underbrace{\|\Delta\|_\infty}_{III}.
\end{align*}

 From the fact that $ \thetas $ is the unique minimizer of $ \Ex \rho(\by_i, \bx_i^T \theta) $ we get \[ \Ex \left[\drho(\by_i,\bx_i^T \thetas) \bx_i\right] =0. \] 
 
 Using 3. in lemma \ref{lemma:high-lebel-decomposition-debiased-lasso} and 7. in the Assumption \ref{assump:debiased-lasso-assumption} we have  $  \|\Theta \bx_i \drho(\by_i,\bx_i)\|_\infty\le  MK_3\sqrt{s^*} $ for some $M>0.$ Using  Bernstein's inequality we get,

\[ \Pr \left(\left| \sum_{ i= 1}^n \Theta_{j,\cdot} \bx_i \drho (\by_i,\bx_i^T\thetas)  \right|>t \right) \le 2\exp\left(- \frac{\frac12t^2}{n\sigma_j^2 + \frac13 \sqrt{s^*} MK_3t }  \right), \] where, $ \sigma_j^2 = \Theta_{j,\cdot} \Ex \left[ \nabla \ell(\thetas)\nabla \ell(\thetas)^T \right] \Theta_{j,\cdot}^T. $ For $ t\le \frac{3n\sigma_j^2}{\sqrt{s^*}MK_3} $ we get sub-Gaussian bound 
\[ \Pr \left(\left| \sum_{ i= 1}^n \Theta_{j,\cdot} \bx_i \drho (\by_i,\bx_i^T\thetas)  \right|>t \right) \le 2\exp\left( -\frac{t^2}{4n\sigma_j^2} \right). \]
By union bound \[ \Pr\left( \|\Theta\nabla \ell(\thetas)\|_\infty > t \right) \le 2d \exp\left( -\frac{nt^2}{4\sigma^2} \right), \] 
where, $ \sigma^2=\max_j \sigma_j^2.  $ Setting $ t = \sigma\sqrt{\frac{12\log d}{n}} $ we get $\Pr \left(\|\Theta\nabla \ell(\thetas)\|_\infty > \sigma\sqrt{\frac{12\log d}{n}}\right) \le \frac2{d^2} .$  Since, for sufficiently large $ n $ we have $ \sigma\sqrt{\frac{12\log d}{n}}\le\frac{3n\sigma_j^2}{\sqrt{s^*}MK_3},  $ such a choice for $ t $ is justified.

Since, $ \left\{ \bx_i \drho (\by_i,\bx_i^T\thetas)\right\}_{i=1}^n $ is bounded and zero mean random vectors, we get the following high probability bound for $ \|\nabla \ell(\thetas)\|_\infty: $
\[ \text{For some } c_4>0,\  \Pr \left( \|\nabla\ell(\thetas) \|_\infty >  c_4\sqrt{\frac{\log d}{n}}\right) \le \frac1{d^2}. \]

Let $ A $ the event that the followings hold:
\begin{enumerate}
	\item $  \|\Delta\|_\infty  >c_1  Ks_0\log d/n, $
	
	\item $ \max_j \|\Thetah_{j,\cdot}-\Theta_{j,\cdot}\|_1 \le c_2  \max\{Ks^*, K^2 s_0\} \sqrt{\log d /n}, $
	
	\item $  \|\Theta\nabla \ell(\thetas)\|_\infty \le \sigma\sqrt{\frac{12\log d}{n}}, $
	\item  $ \|\nabla\ell(\thetas) \|_\infty \le   c_4\sqrt{\frac{\log d}{n}}. $
\end{enumerate}
Then $ \Pr(A) \ge 1 - o(d^{-1}). $ Under the event $ A, $ $ (I) \le \sigma\sqrt{\frac{12\log d}{n}}, \text{ and } (III) \le c_1 Ks_0\log d/n. $ We also notice that $ (II) \le \max_j \|\Thetah_{j,\cdot}-\Theta_{j,\cdot}\|_1\|\nabla\ell(\thetas) \|_\infty \le c_2 c_ 4 \max\{Ks^*, K^2 s_0\} \log d/n.  $

Hence, $ \Pr \left(\|\ttheta^d - \theta^*\|_\infty \le \sigma \sqrt{\frac{12\log d}{n}}  + c \max\{Ks^*, K^2 s_0\} \frac{s^*\log d}{n}\right) \ge 1- o(d^{-1}). $

\end{sproof}

\begin{sproof}{Proof of Result \ref{result:global-uniqueness}}
	We shall prove uniqueness of the global parameter for each co-ordinates separately. Fix $j\in[d].$
	For simplicity of the notation we ignore the index $ j, $ and denote $ I_j, (\thetas_k)_j,\mu_j, $ and, $ \eta_j $ as $ I, \thetas_k,\mu $ and, $ \eta, $ respectively.
	
	For $ x\in\reals $ let us define $ I_x = \{k : \thetas_k\in [x-\eta,x+\eta]\}, $ and  $ N_x = \sum_{k\in I_x} n_k. $

	When  $ I_x =  I, $ we have 
\[ \sum_{k=1}^m \Psi_{\eta}\left(\thetas_{k} - x \right) = (m - |I|)\eta^2 + \sum_{ k \in  I} (\thetas_k-x)^2  \] which	is uniquely minimized at $ x=\mu. $  
	
	Under the case $ I_x \neq I $ we must have $ |\mu -x|>\delta. $ We consider the case $ \delta<|\mu -x|\le 3\delta.$ $ |\mu -x|>\delta $ implies either $ x $ is larger or smaller than  $ \thetas_k $'s which are in $ I $. We assume that $ x>\thetas_k $ for all $ k\in I. $ The other case will follow similarly. In that case, \[ \left(\thetas_{\text{max}}- \thetas_k\right) \le (2\delta)\wedge (x-\thetas_{k}) , \] for all $ k\in I, $ where $ \thetas_{\text{max}} = \max_{j\in I}\thetas_k. $ Since, $ \eta>2\delta, $ we have
	 \[ \sum_{k\in I} \Psi_{\eta}\left(\thetas_{k} - x \right) \ge \sum_{k\in I} \Psi_{\eta}\left(\thetas_{k} -\thetas_{\text{max}}  \right).  \]
	  Notice that, for $ |x-\mu|\le 3\delta $ we have $ \Psi_\eta(\thetas_k-x) = \eta^2 $ whenever $ k\notin I. $ Hence, 
	\[ \sum_{k=1}^m \Psi_{\eta}\left(\thetas_{k} - x \right) \ge \sum_{k=1}^m \Psi_{\eta}\left(\thetas_{k} -\thetas_{\text{max}}  \right) > \sum_{k=1}^m \Psi_{\eta}\left(\thetas_{k} -\mu  \right),  \] where the last inequality follows form the fact that $ I_{x} = I $ for $ x = \thetas_{\text{max}}. $
	
	Now, consider the case $ |x-\mu|>3\delta, $ under which we have $ I_x \cap I = \emptyset. $

	 Then 
	\[ 	\sum_{k=1}^m \Psi_{\eta}\left(\thetas_{k} - x \right)  > 4|I|\delta^2 \] and 
 \[ \sum_{k=1}^m \Psi_{\eta}\left(\thetas_{k} - \mu \right) = (m- |I|)\eta^2 + \sum_{ k \in  I}  (\thetas_{k}-\mu)^2 . \] Since, the range of $ \{\thetas_k\}_{k\in I} $ is less than $ 2\delta, $ we have $ \sum_{ k \in  I}  (\thetas_{k}-\mu)^2 \le \delta^2 |I|.  $ This implies 
 \[ 
 \begin{aligned}
 \sum_{k=1}^m \Psi_{\eta}\left(\thetas_{k} - \mu \right)& \le (m- |I|)\eta^2 +  \delta^2 |I|\\
 &\le 4(m- |I|) \delta^2 + |I|\delta^2\\
 &\le (4m-3|I|)\delta^2\le 4|I|\delta^2 
 \end{aligned}
 \] and hence, $ \mu  $ is the unique minimizer in this case. 
\end{sproof}

\begin{sproof}{Proof of Lemma \ref{lemma:integrated-hp-bound}}

We start with  remark \ref{remark:hp-bound-local} that under the assumption $ (ii) $ for sufficiently large $ n_k $ we have  
\[ \|\ttheta^d_{k} - \thetas_k\|_\infty \le 2\sigma  \sqrt{\log d/n_k}   \] with probability at least $ 1-o(d^{-1}). $
By union bound, with probability $ 1-o(1), $ simultaneously for all $ k $ we have \[ \|\ttheta^d_{k} - \thetas_k\|_\infty \le 2\sigma \sqrt{\log d/n_k}.  \] 

 If $ 2\sigma \sqrt{\log d/n_k} \le \frac14(\eta_j-2\delta)\wedge (\eta_j-\delta_2/2) $ for all $ k, $ then there exists $ \hat \delta  $ and $ \hat \delta_2 $ such that the following holds, \[ 2\delta\le 2\hat\delta <\eta_j<\hat{\delta}_2/2 \le \delta_2/2. \] and assumptions \ref{assump:cluster-parameter} holds with $ \hat{\delta} $ and $ \hat{\delta}_2. $
 Hence, by result \ref{result:global-uniqueness} we have $ (\ttheta_0)_j = \frac{\sum_{ k \in  I_j} (\ttheta_k^d)_j}{|I_j|}.  $ 
 
 Let  $  \gamma^{(j)}_k = \frac{1}{|I_j|}\indicator_{I_j}(k) ,$ where, $ \indicator_A $ is the indicator function over the set $ A. $ From Lemma \ref{lemma:high-lebel-decomposition-debiased-lasso}
 \begin{align*}
 	(\ttheta_0)_j & = \sum_{ k= 1}^m \gamma^{(j)}_k (\ttheta^d_k)_j\\
 	& = \sum_{ k= 1}^m \gamma^{(j)}_k \left((\thetas_k)_j - (\Thetah_k)_{j\cdot} \nabla  \ell_k (\thetas_k ) + (\Delta_k)_j    \right)\\
 	& = (\thetas_0)_j -  \sum_{ k= 1}^m \gamma^{(j)}_k  (\Thetah_k)_{j\cdot} \nabla  \ell_k (\thetas_k ) +  	 \sum_{ k= 1}^m \gamma^{(j)}_k (\Delta_k)_j.
 \end{align*}
  
 Hence,
 \begin{align*}
 	\left|(\ttheta_0)_j-(\thetas_0)_j \right| \le & \left | \sum_{ k= 1}^m \gamma^{(j)}_k  (\Thetah_k)_{j\cdot} \nabla  \ell_k (\thetas_k ) \right |  + \left | \sum_{ k= 1}^m \gamma^{(j)}_k (\Delta_k)_j \right | \\
 	\le & \left | \sum_{ k= 1}^m \gamma^{(j)}_k  (\Theta_k)_{j\cdot} \nabla  \ell_k (\thetas_k ) \right | + \left | \sum_{ k= 1}^m \gamma^{(j)}_k  \left((\Thetah_k)_{j\cdot} -  (\Theta_k)_{j\cdot}\right) \nabla  \ell_k (\thetas_k ) \right |  \\
 	&+ \left | \sum_{ k= 1}^m \gamma^{(j)}_k (\Delta_k)_j \right |\\
 	\le & \underbrace{\left | \sum_{ k= 1}^m \gamma^{(j)}_k  (\Theta_k)_{j\cdot} \nabla  \ell_k (\thetas_k ) \right |}_{I}\\
 	&\quad+ \underbrace{\sum_{ k= 1}^m \gamma^{(j)}_k  \left( \|(\Thetah_k)_{j\cdot} -  (\Theta_k)_{j\cdot}\|_1 \| \nabla  \ell_k (\thetas_k )\|_\infty + |(\Delta_k)_j|   \right)}_{II}.
 \end{align*}

 To bound $ (I) $ we notice that $ \left\{  (\Theta_k)_{j,\cdot} \bx_{ki} \drho(\by_{ki},\bx_{ki}^T\thetas_{k})  \right\}_{i\in [n_k], k\in[m]} $ are zero mean random variables bounded by $ M\sqrt{s^*}>0. $ By Bernstein inequality,
  \[ \Pr\left( \left | \sum_{ k= 1}^m \gamma^{(j)}_k  (\Theta_k)_{j\cdot} \nabla  \ell_k (\thetas_k ) \right |>t   \right) \le 2\exp\left( - \frac{\frac12 t^2}{\sum_{ k= 1}^m \left(\gamma_k^{(j)}\right)^2 \frac{\sigma_{jk}^2}{n_k}  + \frac13MtM\sqrt{s^*} }  \right), \] where $ \sigma_{jk}^2 = \Ex \left[ (\Theta_k)_{j,\cdot} \nabla \ell_k(\thetas_k) \nabla \ell_k(\thetas_k)^T (\Theta_k)_{j,\cdot}^T  \right]. $ 
  Let $ a_{j}^2 = \sum_{ k= 1}^m \left(\gamma_k^{(j)}\right)^2 \frac{\sigma_{jk}^2}{n_k} = \frac1{|I_j|^2}  \sum_{ k= 1}^m \frac{\sigma_{jk}^2}{n_k}. $ For $  t \le \frac{3a_{jk}^2}{M\sqrt{s^*}}$ we get a probability bound with subgaussian tail 
  \[ \Pr\left( \left | \sum_{ k= 1}^m \gamma^{(j)}_k  (\Theta_k)_{j\cdot} \nabla  \ell_k (\thetas_k ) \right |>t   \right) \le 2\exp\left( -\frac{ t^2}{4a_{j}^2}  \right). \] 
  Letting $  t = 2a_{j} \sqrt{3\log d} $ we get \[ \left | \sum_{ k= 1}^m \gamma^{(j)}_k  (\Theta_k)_{j\cdot} \nabla  \ell_k (\thetas_k ) \right | \le 2a_{j} \sqrt{3\log d} \] with probability at least $ 1- d^{-2}. $ Taking union bound over all co-ordinates we get the above bound for each co-ordinates with probability at least $  1- d^{-1}. $

  We shall apply Bennett's concentration inequality \ref{lemma:bennett} on $ (I). $ Fix some $ j\in[d]. $ For $ k\in[m], \ i \in[n_k] $ define $ g_{ik} = \frac{\Theta_{j,\cdot} \nabla\rho(\by_{ki},\bx_{ki}^T \thetas_k)}{MK_3\sqrt{s^*}}, $ and $ a_{ki} = \frac{\gamma^{(j)}_kMK_3\sqrt{s^*}}{n_k}. $ Then $ \Ex g_{ki} = 0 $ and $ \Ex g_{ki}^2 \le \frac{\sigma^2}{M^2K_3^2 s^*} := \delta  $ where, $ \sigma^2 \ge \max_{k \in [m], j\in[d]} \Theta_{j,\cdot} \Ex [\nabla\rho(\by_{ki},\bx_{ki}^T \thetas_k)\nabla\rho(\by_{ki},\bx_{ki}^T \thetas_k)^T] \Theta_{j,\cdot}^T,  $ from assumption (A5). Also, $ |g_{ki}| \le \frac{\|\Theta_{j,\cdot}\|_1 \|\nabla\rho(\by_{ki},\bx_{ki}^T \thetas_k)\|_\infty}{MK_3\sqrt{s^*}} \le 1. $  Let us define $ a \coloneqq (a_{ki}, k \in [m], \ i \in [n_k] ) .$ Now we apply Bennett's inequality \ref{lemma:bennett} on $ \sum_{ k= 1}^m \gamma^{(j)}_k  (\Theta_k)_{j\cdot} \nabla  \ell_k (\thetas_k ) = \sum_{ k =1}^m \sum_{ i= 1}^{n_k}a_{ki}g_{ki}.  $ We notice that for $ t \le \frac{\delta\|a\|_2^2}{\|a\|_\infty }e^2 $ we have the following bound:
  \[ \Pr \left( \left| \sum_{ k =1}^m \sum_{ i= 1}^{n_k}a_{ki}g_{ki} \right| > t  \right) \le 2\exp \left(-\frac{t^2}{2\delta\|a\|_2^2e^2}\right). \] For $ t = 2 \sqrt{\delta \log d} \|a\|_2 e $ we get \[ \Pr \left( \left| \sum_{ k =1}^m \sum_{ i= 1}^{n_k}a_{ki}g_{ki} \right| > 2 \sqrt{\delta \log d} \|a\|_2 e  \right) \le \frac2{d^2}. \] We need to confirm that such a choice of $ t $ is valid, \ie, \ $  2 \sqrt{\delta \log d} \|a\|_2 \le \frac{\delta\|a\|_2^2}{\|a\|_\infty }e^2 $ or $ \|a\|_\infty \sqrt{\log d} \le c \sqrt{s^*}\|a\|_2 $ for some $ c  $ independent of $ m, d $ and $ n_k , \in [m]. $ We notice that 
  \[ \|a\|_2 = \frac{MK_3\sqrt{s^*}}{|I_j|}\sqrt{\sum_{ k \in  I_j}\frac1{n_k}} \ge \frac{1}{\sqrt{|I_j| n_{\max}}}, \text{and } \|a\|_\infty \le \frac{MK_3\sqrt{s^*}}{n_{\min}|I_j|} \]
  and hence, \[ \frac{\sqrt{s^*}\|a\|_2}{\|a\|_\infty \sqrt{\log d}} \ge \sqrt{\frac{s^*n_{\min}^2|I_j|}{n_{\max}\log d}} \ge \sqrt{\frac{n_{\min}^2}{{s^*}^2n_{\max}\log d}}. \] From assumption $ (iii) $ we get that $ \sqrt{\frac{n_{\min}^2}{{s^*}^2n_{\max}\log d}} $ asymptotically goes to infinity. Hence, such a choice of $ t  $ is valid.

  We notice that $ \|a\|_2 $ is dependent of $ j. $ We define $ b_j =: 2 \sqrt{\delta} \|a\|_2 e \le \frac{2e\sigma}{|I_j|}\sqrt{\sum_{ k \in  I_j}\frac1{n_k}}. $
  
  To get a bound for $  (II) $ we see that $ \max_j \|(\Thetah_k)_{j\cdot} -  (\Theta_k)_{j\cdot}\|_1 \lesssim  \max\{K s_k^*, K^2 s_{k,0}\} \sqrt{\log d/n_k } , $ and $  \|\Delta_k\|_\infty \lesssim K s_{k,0}\log d /n_k.   $ We  also notice that $ \{\bx_{ki} \drho (\by_{ki}, \bx_{ki}^T \thetas_k)  \}_{i \in [n_k], k\in [m]} $ are mean zero bounded random vectors. Hence, by Bernstein inequality we can show that for sufficiently large $ n_k $'s we can get a bound \[ \|\nabla \ell_k(\thetas_k)\|_\infty \lesssim \sqrt{\frac{\log d}{n_k}} \] with probability at least $  1- d^{-2}. $ Hence, by union bound we get \[ \|(\Thetah_k)_{j\cdot} -  (\Theta_k)_{j\cdot}\|_1 \| \nabla  \ell_k (\thetas_k )\|_\infty + |(\Delta_k)_j| \lesssim \frac{\max\{K s_{\max}^*, K^2 s_{0,\max}\}\log d}{n_{\text{min}}} \text{ for all }j,k  \] with probability at least $  1-o(1). $
  
  Again by union bound, we get 
  \begin{equation}
  	|(\ttheta_0)_j - (\thetas_0)_j | \le b_j \sqrt{\log d} + C \frac{\max\{K s_{\max}^*, K^2 s_{0,\max}\}\log d}{n_{\text{min}}}, \text{ for all } j
  	\label{eq:bound-theta-0}
  \end{equation}
   with probability at least $ 1-o(1), $ where $ C>0 $ is some constant. 
   Since, $b_j \ge \sigma' \sqrt{\frac{1}{mn_{\max}}},$ for some $\sigma',$ from assumption $(iv)$ we get $\frac{\max\{K s_{\max}^*, K^2 s_{0,\max}\}\log d}{n_{\min}} \lesssim b_j \sqrt{\log d}. $ Hence, for sufficiently large $\frac{n^2_{\min}}{n_{\max}}$,  the above high probability bound reduces to
  \[ |(\ttheta_0)_j - (\thetas_0)_j | \le b_j \sqrt{\log d} \] with probability at least $ 1-o(1). $ 
  
  Under the assumption $ (iii) $ we have $ b_j \le 2e\sigma\sqrt{\frac{1}{|I_j|^2}\sum_{ k \in  I_j}\frac1{n_k}}, $ which gives us the first result. 
  
  For the second inequality we get a bound for each $ \ttheta_{k}^d $ of the form:
  \[ \|\ttheta_k^d - \thetas_{k}\|_\infty \le 2 \sigma\sqrt{\frac{2\log d}{n_k}} + C_k \frac{\max\{K s_{k}^*, K^2 s_{k,0}\}\log d}{n_k}, \text{ for each }k \] with probability at least $ 1-o(1). $ 
  Form the above bound and  \eqref{eq:bound-theta-0} we get \[ 
  \begin{aligned}
  \|\tdelta_k-\deltas_k\|_\infty  & \le 2\sigma \sqrt{\frac{2\log d}{n_k}} + C_k \frac{\max\{K s_{\max}^*, K^2 s_{0,\max}\}\log d}{n_k}\\
  & ~~~+ 2e\sigma\sqrt{\frac{\log d}{|I_j|^2}\sum_{ k \in  I_j}\frac1{n_k}} + C \frac{\max\{K s_{k}^*, K^2 s_{k,0}\}\log d}{n_{\text{min}}}.
  \end{aligned}
   \] hence, for sufficiently large $ n_k $'s, we have \[ \|\tdelta_k-\deltas_k\|_\infty \le 4\sigma \sqrt{\frac{\log d}{n_k}} \text{ for all }k,\] with probability at least $ 1-o(1). $

\end{sproof}

\begin{result}
	For each $ k\in[m] $ let $ \|\widetilde\theta_{k}-\thetas_k\|_\infty \le \xi,  $ where $ \xi < \min_j \frac14(\eta_j - 2\delta_j)\wedge (\delta_{2j}/2 -\eta_j). $ Then the assumptions \ref{assump:cluster-parameter} are satisfied for $ \{\widetilde\theta_{k}\}_{k=1}^m $ for $ \{\widetilde{\delta}_j = \delta_j + \xi/2 \}_{j =1}^d $, $ \{ \widetilde{\delta}_{2j} = \delta_{2j} -2\xi  \}_{j=1}^d $ and $ \{\eta_j\}_{j = 1}^d. $
\end{result}

\begin{proof}
	Let $ j \in [d]. $ Consider the same $ I_j $ as in assumption \ref{assump:cluster-parameter}, $ (i). $ Let $ \widetilde{\mu} = \frac1{I_j}\sum_{ k \in  I_j} \widetilde{\theta}_k. $ Notice that for $ \|\widetilde\theta_{k}-\thetas_k\|_\infty \le \xi  $ we have $ |\widetilde{\mu}-\mu|\le \xi. $   We further notice that $ 2(\delta_j + 2\xi) < \eta_j < \frac{\delta_{2j}-2\xi}2. $
	This implies assumption  \ref{assump:cluster-parameter}, $ (ii) $ and $ (iii) $ are satisfied with $ \widetilde{\delta}_j = \delta_j + 2\xi $ and $ \widetilde{\delta}_{2j} = \delta_{2j} -2\xi. $
\end{proof}

\begin{result}[Bernstein's inequality]
	Let $ \{\bx_i\}_{i=1}^n $ be independent random variables with $ \Ex \bx_i = 0,\ \Ex \bx_i ^2 \le \sigma^2 . $ Suppose for some $ \lambda_0>0 $ and $ M>0 $ the following holds for all $ i \in [n]: $
	\[ \Ex \exp (\lambda_0|\bx_i|) \le M. \] Then for any $ a \in \reals^n $ and $ t \le \frac{M\sigma^2 \|a\|_2^2\lambda_0 }{\|a\|_\infty} $ the following holds:\[ \Pr\left(\left| \sum_{ i= 1}^n a_i\bx_i  \right| > t  \right) \le 2 \exp\left(- \frac{t^2}{2\|a\|_2^2 \sigma^2 M }  \right). \]
	\label{result:bernstein}
\end{result}

\begin{lemma}[Bennett's Inequality]
	Let $X_1,X_2,\ldots X_n$ be independent random variables such that (1) $\Ex X_i = 0,$ (2) $\Ex X_i^2 \le  \delta ,$ and (3) $|X_i|\le 1.$ Then for any $a_1,a_2,\ldots ,a_n\in \reals$ and for any $t>0$
	\[
	\Pr \left(\left|\sum_{i=1}^n a_iX_i\right|>t\right) \le 
	\begin{cases}
	2\exp \left(-\frac{t^2}{2\delta\|a\|_2^2e^2}\right), &\text{if}\ t\le t^*\\
	2\exp \left(-\frac{t}{4\|a\|_\infty}\log\left(\frac{t\|a\|_\infty}{\delta\|a\|_2^2}\right)\right), &\text{if}\ t>t^*
	\end{cases}
	\]
	where $t^* = \frac{\delta\|a\|_2^2}{\|a\|_\infty }e^2.$
	\label{lemma:bennett}
\end{lemma}

\begin{proof}
	Without loss of generality assume $\|a\|_\infty=1.$ Let $\lambda>0.$ Set $Y=\sum_{i=1}^na_iX_i.$ Then 
	\[
	\Ex \exp(\lambda Y)=\prod_{i=1}^{n} \Ex \exp (\lambda a_i X_i).
	\]
	For any $y\in\reals,$ $e^y\le 1+y+\frac{y^2}2 e^{|y|}.$
	
	Hence
	\begin{align*}
	\Ex \exp (\lambda a_i X_i)&\le 1+ \Ex \frac{\lambda^2 a_i^2 X_i^2}{2} e^{|\lambda a_i X_i|}\\
	&\le 1+\frac{\lambda^2 a_i^2 \delta}{2}e^\lambda\\
	&\le \exp\left(\frac{\lambda^2 a_i^2 \delta}{2}e^\lambda\right)
	\end{align*}
	Therefor $\Ex \exp(\lambda Y)\le \exp\left(\frac{\lambda^2 \|a\|_2^2 \delta}{2}e^\lambda\right).
	$

	By Markov's inequality 
	\[
	\Pr(Y>t) \le \exp\left(\frac{\lambda^2 \|a\|_2^2 \delta}{2}e^\lambda-\lambda t\right).
	\]
	We have to minimize $\phi(\lambda)=\lambda t-\frac{\lambda^2 \|a\|_2^2 \delta}{2}e^\lambda.$
	
	(1) Consider the case $\lambda\le 2 .$ Then $\phi(\lambda)\ge \lambda t-\frac{\lambda^2 \|a\|_2^2 \delta}{2}e^2.$ minimization with respect to $ \lambda $ gives us $ \phi(\lambda)\ge \lambda t-\frac{\lambda^2 \|a\|_2^2 \delta}{2}e^2 \ge \frac{t^2}{2\|a\|_2^2 \delta e^2}. $ Since the minimum attains at $ \lambda = \frac{t}{\|a\|_2^2 \delta e^2}, $ under the condition $ \lambda \le 2 $  we have $ t \le \delta \|a\|_2^2 2e^2. $ Hence, for $ t \le \delta \|a\|_2^2 2e^2 $ we get the upper bound \[ \Pr(Y>t) \le \exp\left(- \frac{t^2}{2\|a\|_2^2 \delta e^2} \right). \]

	(2) Consider the case of $\lambda>1.$ This implies $ \lambda<e^\lambda.$ Then $\phi(\lambda)\ge \lambda t-\frac{\lambda \|a\|_2^2 \delta}{2}e^{2\lambda}=\lambda\left(t-\frac{\|a\|_2^2 \delta}{2}e^{2\lambda}\right).$ Choose $\lambda_2$ such that $\frac{\|a\|_2^2 \delta}{2}e^{2\lambda_2}=\frac t2.$ Then $ \phi(\lambda)\ge \lambda_2 t/2 = \frac t4 \log\left(\frac t{\|a\|_2^2 \delta}  \right). $ Also from  $ \lambda_2>1 $ we get $ \frac12 \log\left(\frac t{\|a\|_2^2 \delta}  \right) > 1  $ which gives us $ t \ge \|a\|_2^2 \delta e^{2}. $ Hence, for $ t \ge \|a\|_2^2 \delta e^{2} $ we get the upper bound \[  \Pr(Y>t) \le \exp\left(- \frac t4 \log\left(\frac t{\|a\|_2^2 \delta}  \right) \right).\]
	
	So, there is an overlap in the interval $ (\delta \|a\|^2 e^2, 2\delta \|a\|_2^2 e^2) $ under which we get both the tails. We just choose $ t^* = \delta \|a\|_2^2 e^2. $
\end{proof}

\section{A weighted data integration}
\label{sec:weighted-mrlasso}
In this section we consider a setup of weighted data integration where the weight  $w = (w_1, \dots, w_m)$ ($\sum_{k=1}^m w_k = 1$) are applied in integration step the debiased lasso estimators: 
\begin{equation}
    \textstyle 
    	\left(\ttheta_0^w\right)_j = \argmin_{x} \sum_{k=1}^{m}  w_k\Psi_{\eta_j}\left(\left(\ttheta^d_k\right)_j-x\right).
	\label{eq:averaged-debiased-lasso-wt}
\end{equation}
The corresponding global parameter $(\theta_0^{w, *})$ is identified via weighted re-descending loss: 
\begin{equation}
    \textstyle
    	\left(\theta_0^* \right)_j = \argmin_{x\in\reals} \sum_{ k= 1}^m w_k\Psi_{\eta_j} (\left(\theta_{k}^*\right)_j-x),
 	\label{eq:globalParameter-nk}
\end{equation} and the unique identification of such global parameter is guaranteed by a similar set of structural conditions as in Assumption \ref{assump:cluster-parameter}. We state the conditions below:  

\begin{assumption}
\begin{enumerate}[label = (P\arabic*)]
	\item 
		Let $ I_j $ be the set of indices for $ (\thetas_{k})_j $'s which are considered as non-outliers. We assume $(\sum_{k \in I_j} w_k) \ge 4/7$.
		\item  Let $ \mu_j = \big(\sum_{ k \in  I_j}w_k(\thetas_{k})_j\big)/(\sum_{ k \in  I_j}w_k) . $ Let $  \delta  $ be the smallest positive real number such that $ (\thetas_k)_j \in [\mu_j-\delta,\mu_j+\delta] $ for all $ k\in I_j. $
		 We assume that none of the  $ (\thetas_k)_j $'s are in the intervals $ [\mu_j-5\delta,\mu_j-\delta) $ or $ (\mu_j+\delta,\mu_j+5\delta]. $ 
		\item Let $ \delta_2 = \min_{k_1\in I_j, k_2\notin I_j} |(\thetas_{k_1})_j - (\thetas_{k_2})_j|. $ Clearly, $ 4\delta<\delta_2. $ We choose $\eta_j$ such that $ 2\delta < \eta_j <\delta_2/2. $
\end{enumerate}
\label{assump:cluster-parameter-nk}
\end{assumption}

\begin{result}
	Under the conditions (P1)-(P3) in Assumption \ref{assump:cluster-parameter-nk}, the objective functions $ \sum_{ k= 1}^m w_k\Psi_{\eta_j}\left((\thetas_{k})_j - \theta \right) $ is uniquely minimized at  $ (\theta_0^{*, w})_j = \mu_j, $ for all $j.$
	\label{result:global-uniqueness-nk}
\end{result}

\begin{proof}
This proof is exactly same as the proof of Result \ref{result:global-uniqueness} if we assume that there are $w_k$ proportitions of $\thetas_k$.
\end{proof}

\begin{theorem}
Let the followings hold:
\begin{enumerate}[label=(\roman*)]
	\item For any $ j, $ $ \left\{(\thetas_{k})_j\right\}_{k=1}^n $ satisfy the assumption \ref{assump:cluster-parameter-nk}.
	\item The datasets $ \left\{\cD_k\right\}_{k=1}^m $ satisfy  the Assumption \ref{assump:debiased-lasso-assumption} uniformly over $ k $
	\item Let  $s_{k, 0}$ be sparsity for $\theta_k^*$ and $s_{k}^*$ being the maximum sparsity for the rows of the inverse of $\Ex[\ell_k(\theta_k^*)]$.  Define $s_{0, \max} = \max_k s_{k, 0}$,  $s_{\max}^* = \max_k s_k^*$ and $s_{\operatorname{eff}} = \max\{K s_{\max}^*, K^2 s_{0,\max}\}$.  We assume $m = o\left( \frac{n^2_{\min}}{s_{\operatorname{eff}}^2n_{\max}\log d }\right),$ where, $ n_{\max} = \max_{k \in [m]} n_k, $ and $n_{\min} = \min_{k \in [m]} n_k.$ 
\end{enumerate}

Then for   sufficiently large $ n_k $ we have the following bound for the co-ordinates of $ \ttheta_0 $ and $ \ell_\infty $ bound for $ \tdelta_k $:
\begin{equation} \textstyle
\begin{aligned}
 \left| (\ttheta_0^w)_j - (\theta_0^{w, *})_j \right| \le  4\sigma\sqrt{\frac{\log d}{(\sum_{k \in I_j}w_k)^2}\sum_{ k \in  I_j}\frac{w_k^2}{n_k}}, \text{ for all }j,\\
 \text{ and } \| \tdelta_k^w - \delta_k^{w, *} \|_\infty  \le  4\sigma\sqrt{\frac{\log d}{n_k}}, \text{ for all }k,\\
 \tdelta^w_k = \ttheta_k^d - \ttheta^w_0,~~ \delta_k^{w, *} 
 = \theta_k^{*} - \theta_0^{w, *} 
 \end{aligned}
\end{equation} with probability at least $ 1-o(1). $
	\label{eq:averaged-debiased-lasso-wt}


\label{th:integrated-hp-bound-nk}
\end{theorem}

\begin{proof}
The proof is similar to the proof of Lemma \ref{lemma:integrated-hp-bound} if we replace $  \gamma^{(j)}_k $ by $  \frac{w_k}{\sum_{k\in I_j}w_k}\indicator_{I_j}(k).$
\end{proof}

\section{Computational results: generalized linear model}

In this section we perform synthetic experiments similar to Section \ref{sec:computationalResults} on binary classification setup. The covariates are $d = 100$ dimensional and again generated from AR(1) model with correlation $\rho = 0.9$ and variance 0.25 ($x = (x_1, \dots. x_d)$, $x_1 \sim N(0, 0.25)$ and $x_j = \rho x_{j - 1} + \sqrt{1 - \rho^2} \eps_j$, $\eps_j \sim N(0, 0.25)$), and the response $y_k$ for $k$-th dataset is generated as:
\[
y_{k,i} \sim \text{Bernoulli} (\sigma(\bx_{k, i}^\top \beta_k^*)), ~~ \sigma(t) = \frac{e^t}{e^t + 1}\,.
\] We preform the debiased lasso on logistic regression to get the local estimates and keep rest of the setup exactly same as before. 

In Figure \ref{fig:l2-glm} compare the $\ell_2$-error of global parameter estimates with corresponding global parameters (note that they are different). The behaviors are similar to the linear models. For a fair comparison with the same baseline, we examine the prediction performance of the global parameters on new test data (not seen before), which one can find in Figure \ref{fig:error.glm}.
The results for the logistic regression model have similar behavior as the linear regression model.

\begin{figure}
	\centering
		\includegraphics[width=\textwidth]{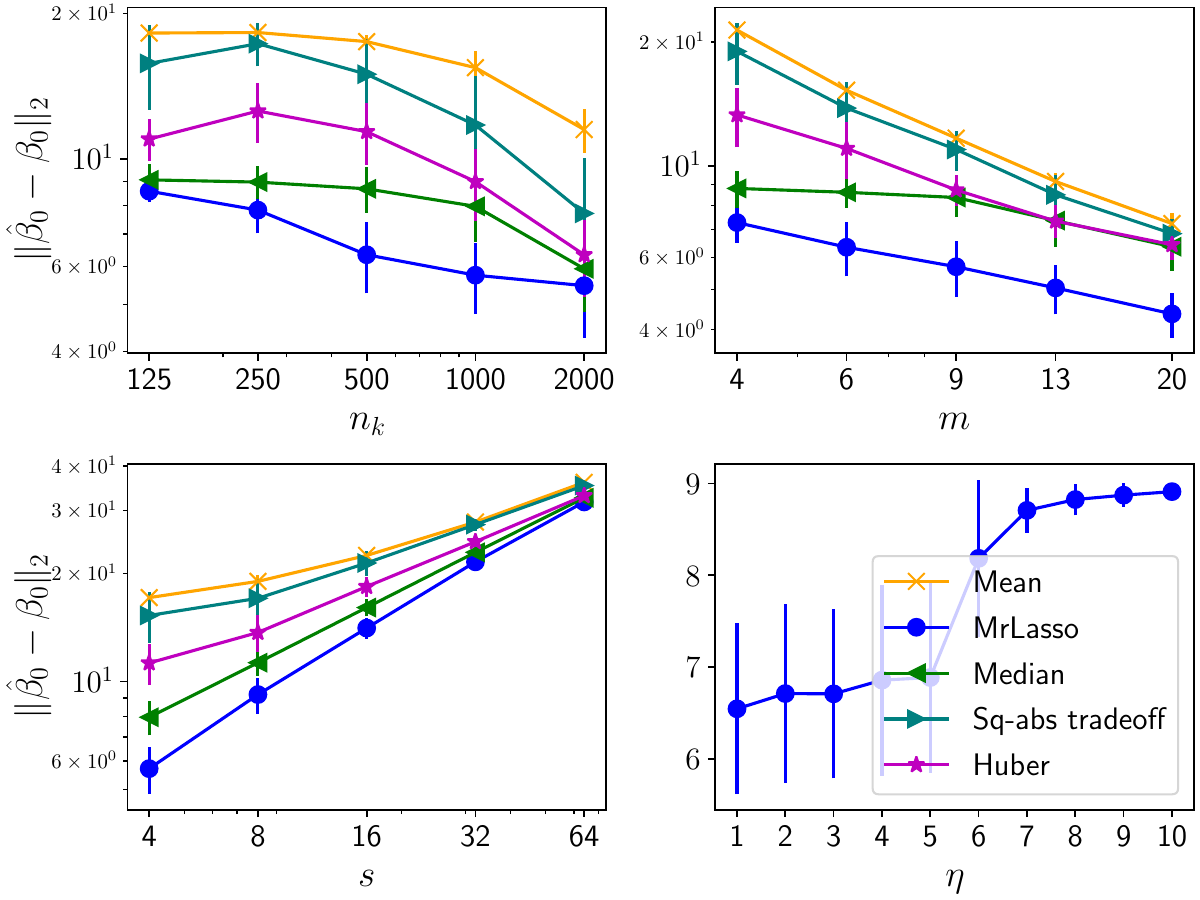}
	\caption{\textit{Upper left}: Errorbar plots for $\ell_2$-error in estimation of $\beta_0$ using different global estimators for varying sample sizes in each datasets ($n_k$) with $m = 5$ and $s = 5.$  \textit{Upper right}: For varying $m$ with $n_k = 200$ and $s = 5$. \textit{Lower left}: For different $s$ with $n_k = 200$ and $m = 5$.   \textit{Lower right}: For different $\eta$ with $n_k = 200$ and $s = 5$ and $m = 5.$}
	\label{fig:l2-glm}
	
\end{figure}

\begin{figure}
	\centering
		\includegraphics[width=\textwidth]{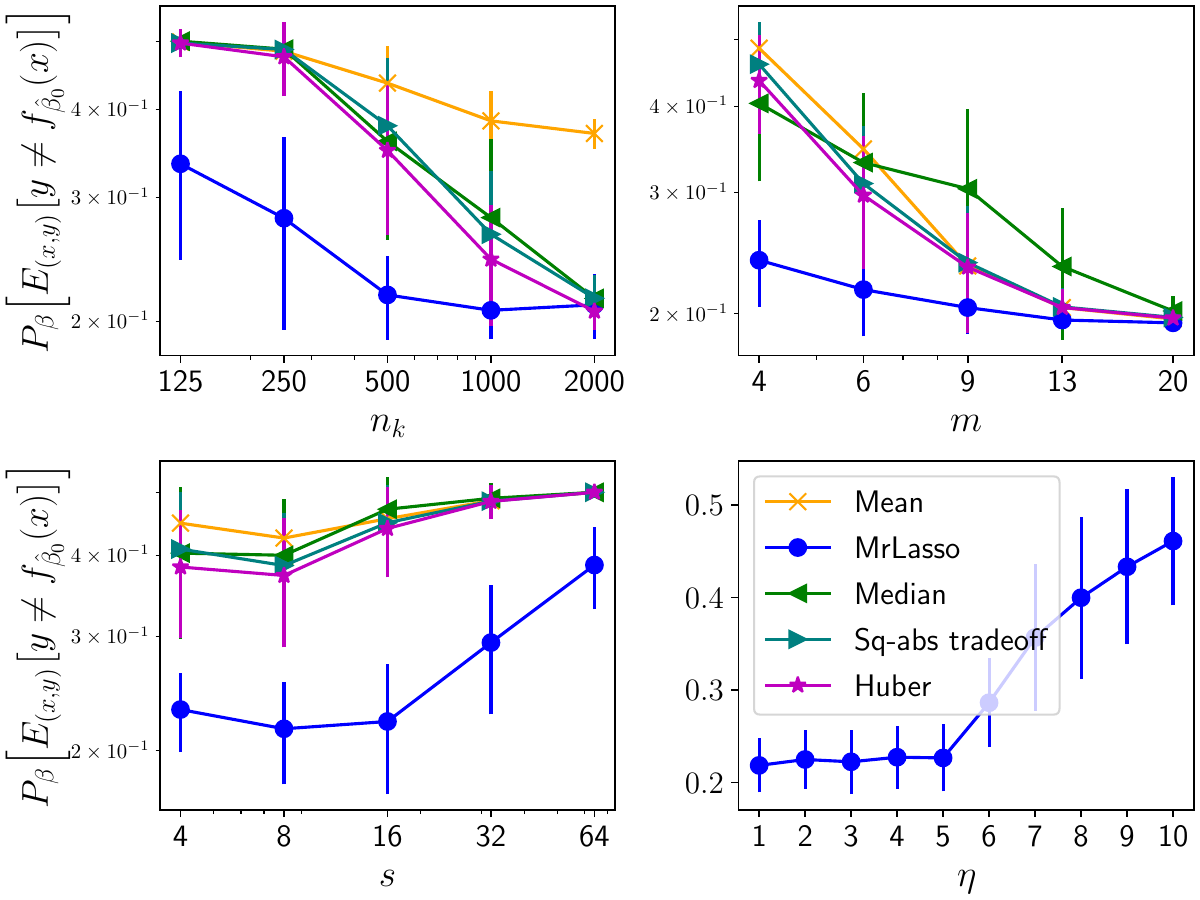}
	\caption{\textit{Upper left}: Errorbar plots for prediction errors in test data  using global estimators for varying sample sizes in each datasets ($n_k$) with $m = 5$ and $s = 5.$  \textit{Upper right}: For varying $m$ with $n_k = 200$ and $s = 5$. \textit{Lower left}: For different $s$ with $n_k = 200$ and $m = 5$.   \textit{Lower right}: For different $\eta$ with $n_k = 200$ and $s = 5$ and $m = 5.$ }
	\label{fig:error.glm}
	
\end{figure}

\section{Supplementary details for cancer cell line study}
\label{sec:supp-ccle}
In the following section we provide the detail about cancer cell line dataset. 

\subsection{Data source and preprocessing}
The Cancer Cell Line Encyclopedia is a database on gene expression, genotype and drug sensitivity data. As covariates, we use the RNAseq TPM gene expression data \cite{depmap_2021} for just protein coding genes (DepMap 21Q2 Public release) and the pharmacologic profiles for 24 anticancer drugs \cite{cancer2015pharmacogenomic} across 504 across lines as the responses. The cell-lines in the gene expression data and the drug response data are matched using the pairs of cell line name and depmap id from the secondary screen dose response curve parameters file \cite{corsello2019non} (secondary-screen-dose-response-curve.csv). These data files are publicly available in DepMap portal\footnote{Depmap portal: https://depmap.org/portal/}

The gene expression data has  genetic expression  for $d = 19177$ genes across 1379 cell lines. After matching them with the drug response data we only keep the  cell lines for which the dose-response curve was fitted using sigmoid function. We use the area above dose-response curve (ActArea) as the drug response measure. Finally we end up with $482$ cell lines in total. 
The codes are available in \href{https://github.com/smaityumich/MrLasso}{https://github.com/smaityumich/MrLasso}.

\subsection{Drug-response prediction plots}
Next we present prediction error plots and gene selections for the drugs 17-AAG, Irinotecan, AZD0530, AEW541, ZD-6474, TKI258, RAF265 and PF2341066. 

\begin{figure}
 \centering
 \includegraphics[width=\textwidth]{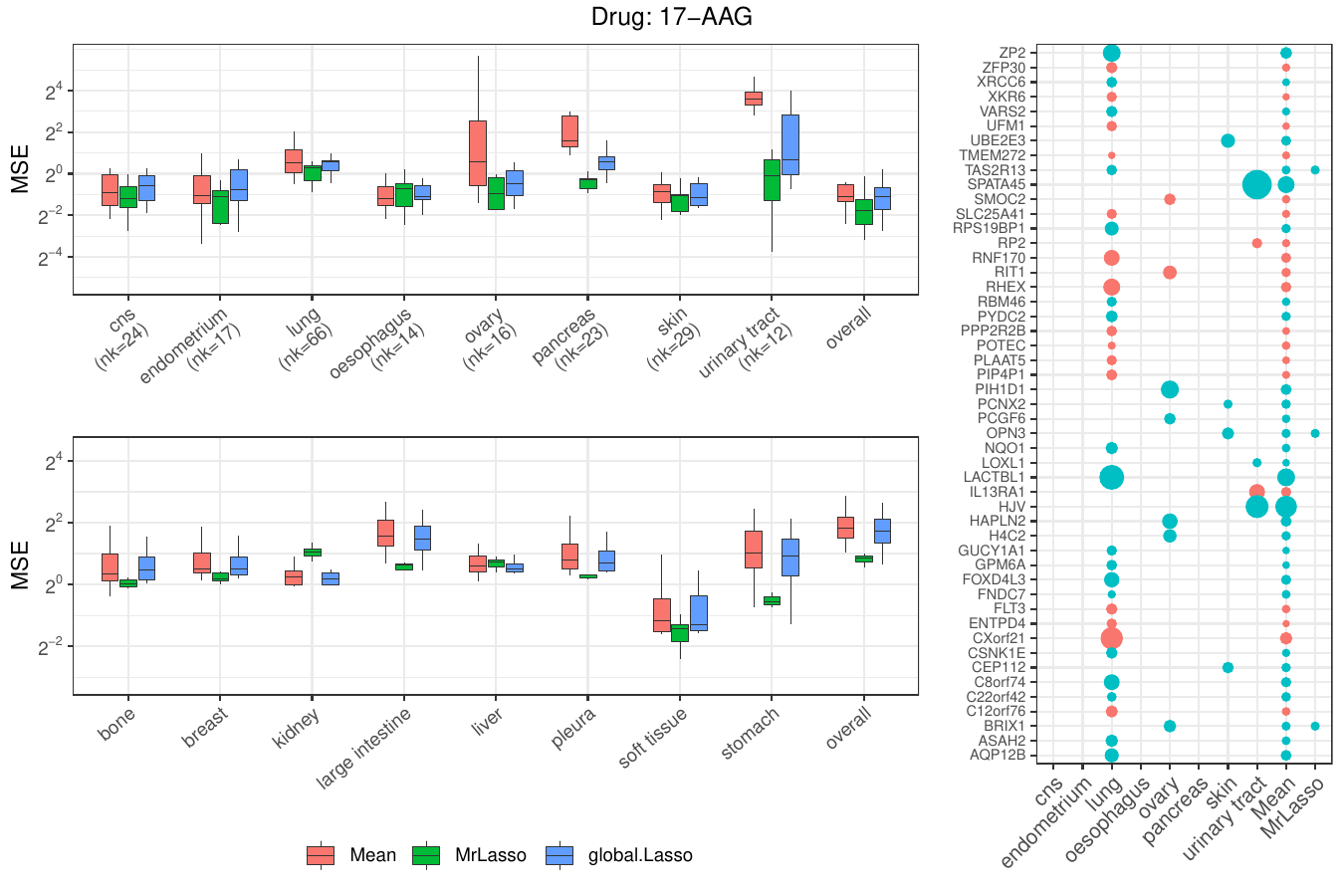}

 \caption{\textit{Left:} Drug-response prediction mean squared errors (MSEs) for  the most frequent cancer types (upper left) and rare cancer types (lower left panel) for the drug 17-AAG.  \textit{Right:} Coefficient plot for the selected genes in lasso estimates for most frequent cancer types and ADELE and Mr Lasso.}
 \label{fig:ccle-17-AAG}
\end{figure}

\begin{figure}
 \centering
 \includegraphics[width=\textwidth]{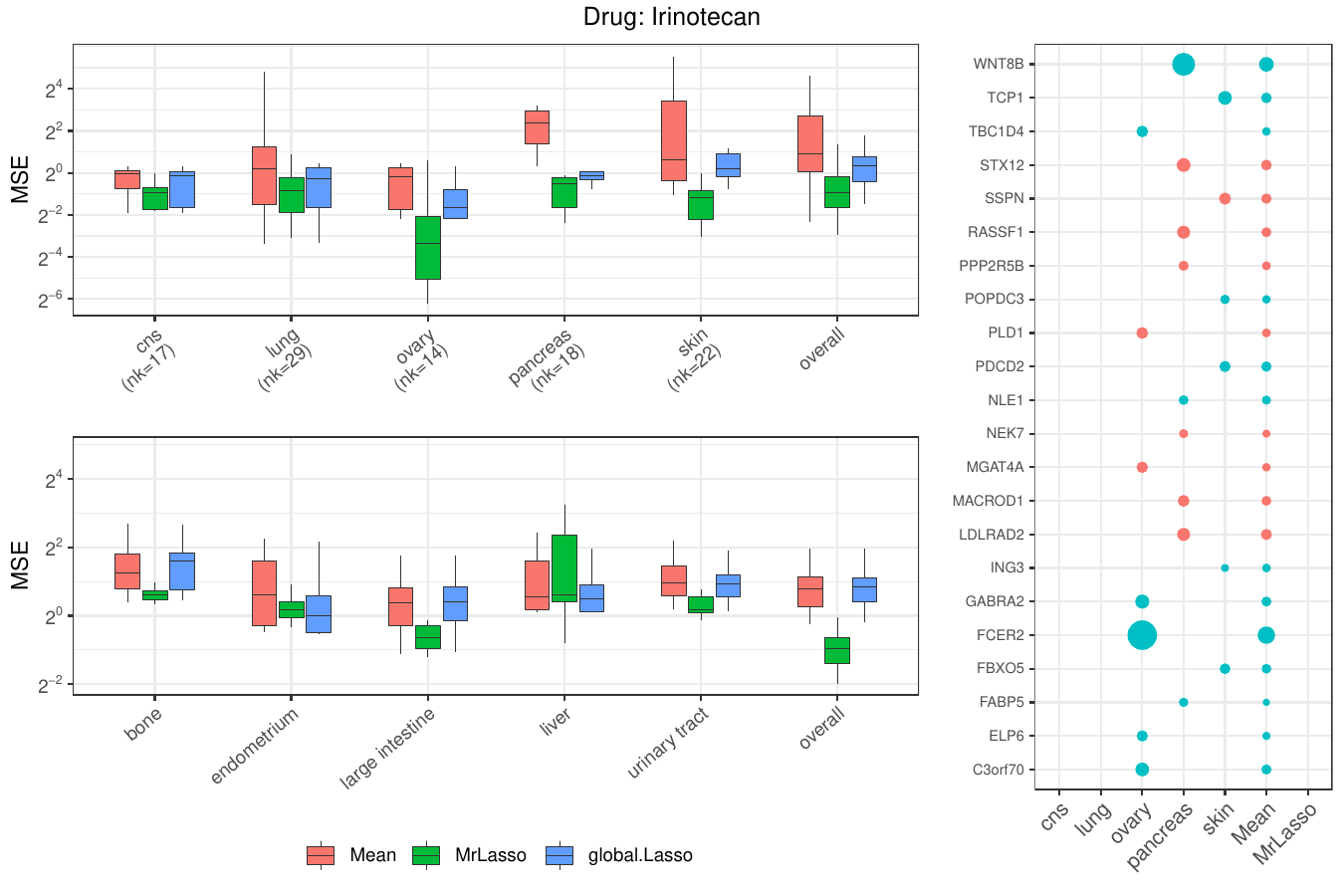}

 \caption{\textit{Left:} Drug-response prediction mean squared errors (MSEs) for  the most frequent cancer types (upper left) and rare cancer types (lower left panel) for the drug Irinotecan.  \textit{Right:} Coefficient plot for the selected genes in lasso estimates for most frequent cancer types and ADELE and Mr Lasso.}
 \label{fig:ccle-Irinotecan}
\end{figure}

\begin{figure}
 \centering
 \includegraphics[width=\textwidth]{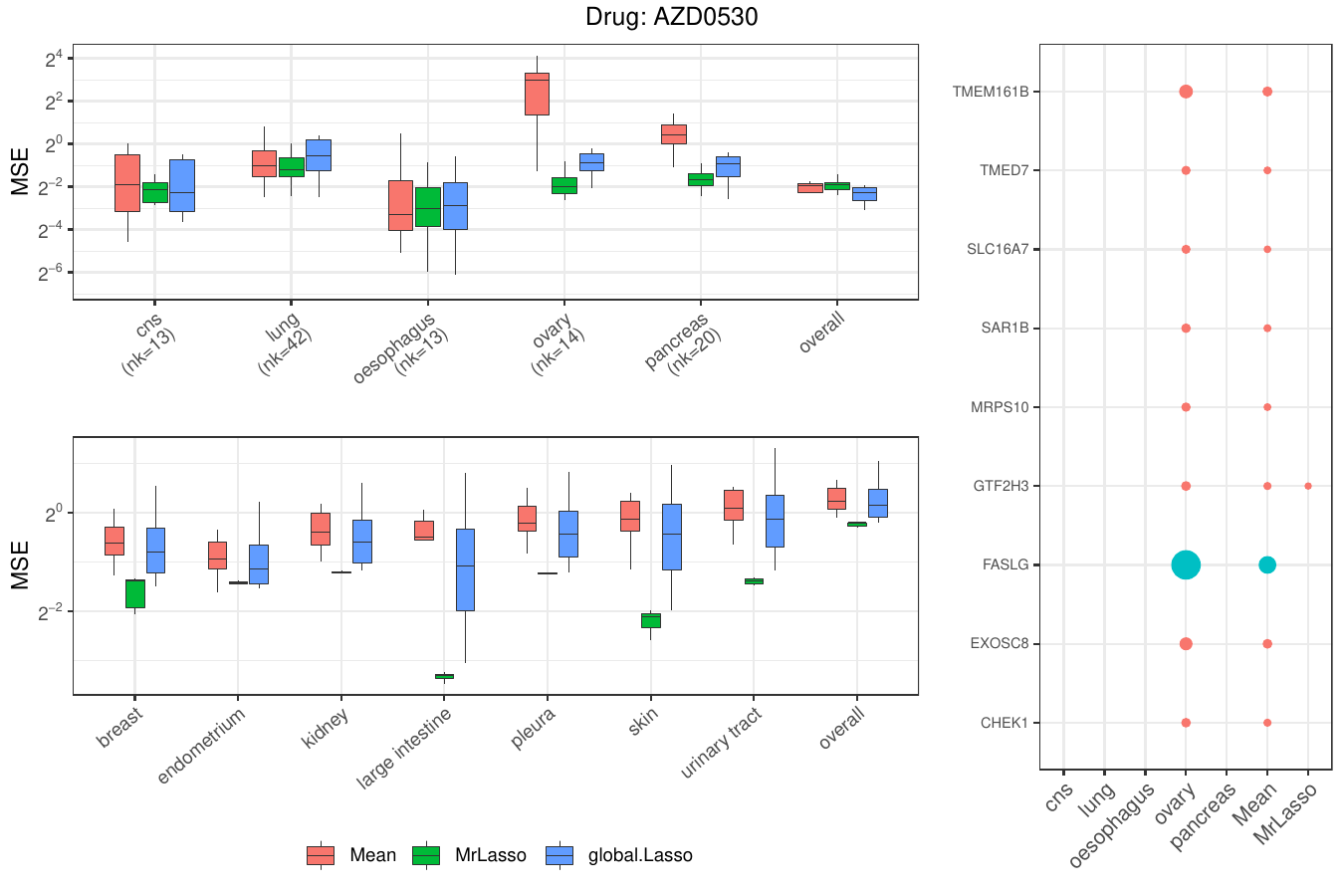}

 \caption{\textit{Left:} Drug-response prediction mean squared errors (MSEs) for  the most frequent cancer types (upper left) and rare cancer types (lower left panel) for the drug AZD0530.  \textit{Right:} Coefficient plot for the selected genes in lasso estimates for most frequent cancer types and ADELE and Mr Lasso.}
 \label{fig:ccle-AZD0530}
\end{figure}

\begin{figure}
 \centering
 \includegraphics[width=\textwidth]{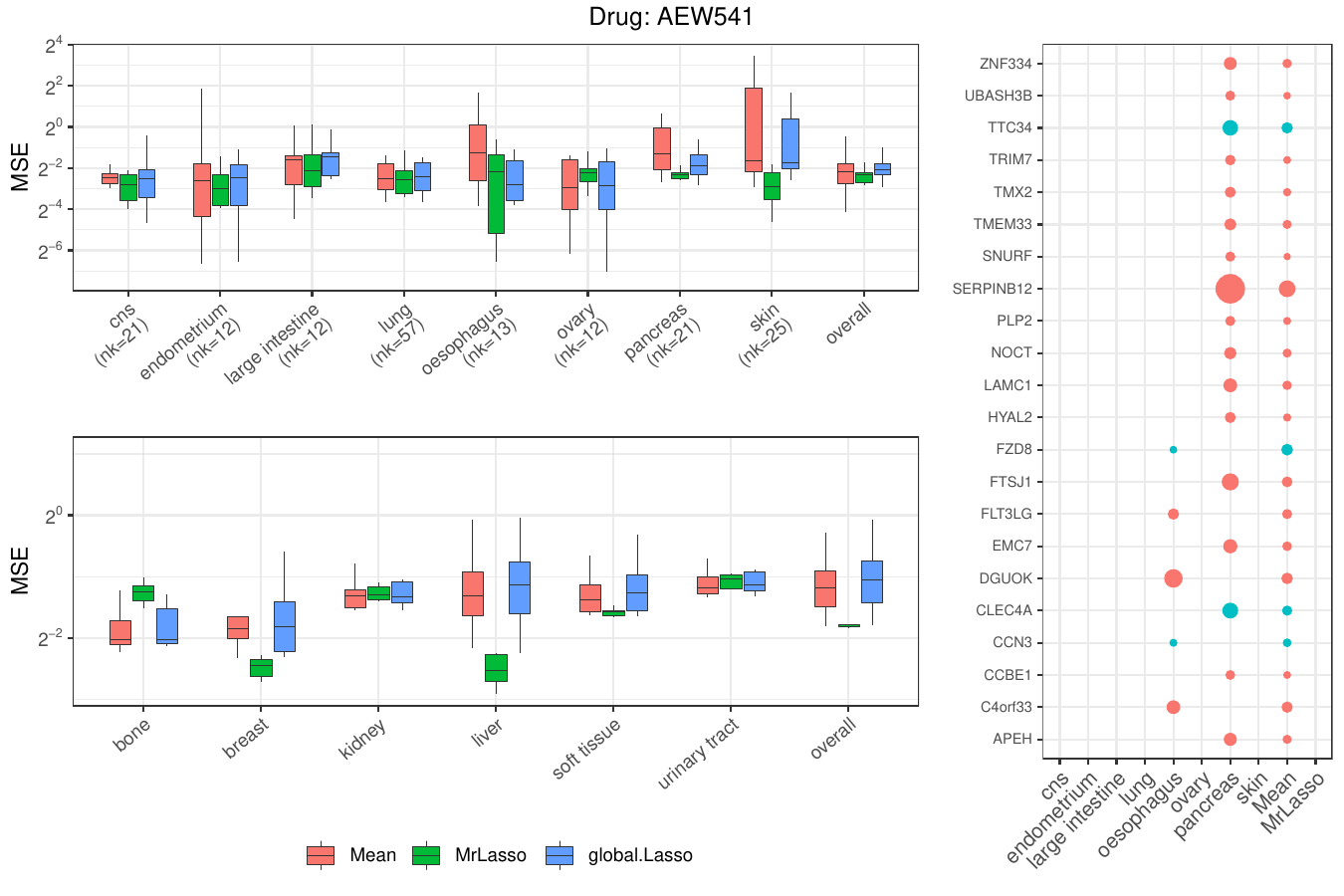}

 \caption{\textit{Left:} Drug-response prediction mean squared errors (MSEs) for  the most frequent cancer types (upper left) and rare cancer types (lower left panel) for the drug AEW541.  \textit{Right:} Coefficient plot for the selected genes in lasso estimates for most frequent cancer types and ADELE and Mr Lasso.}
 \label{fig:ccle-AEW541}
\end{figure}

\begin{figure}
 \centering
 \includegraphics[width=\textwidth]{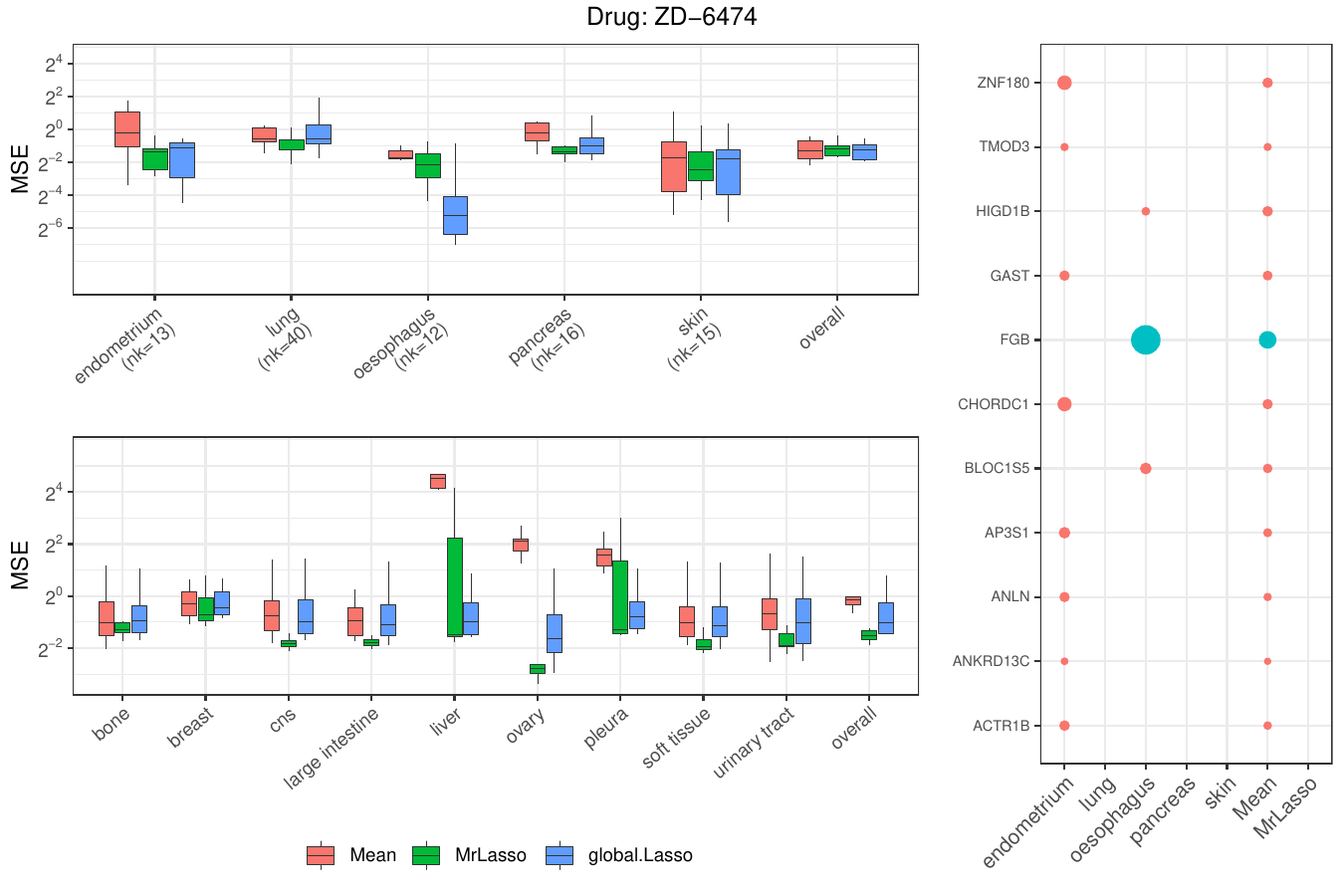}

 \caption{\textit{Left:} Drug-response prediction mean squared errors (MSEs) for  the most frequent cancer types (upper left) and rare cancer types (lower left panel) for the drug ZD-6474.  \textit{Right:} Coefficient plot for the selected genes in lasso estimates for most frequent cancer types and ADELE and Mr Lasso.}
 \label{fig:ccle-ZD-6474}
\end{figure}

\begin{figure}
 \centering
 \includegraphics[width=\textwidth]{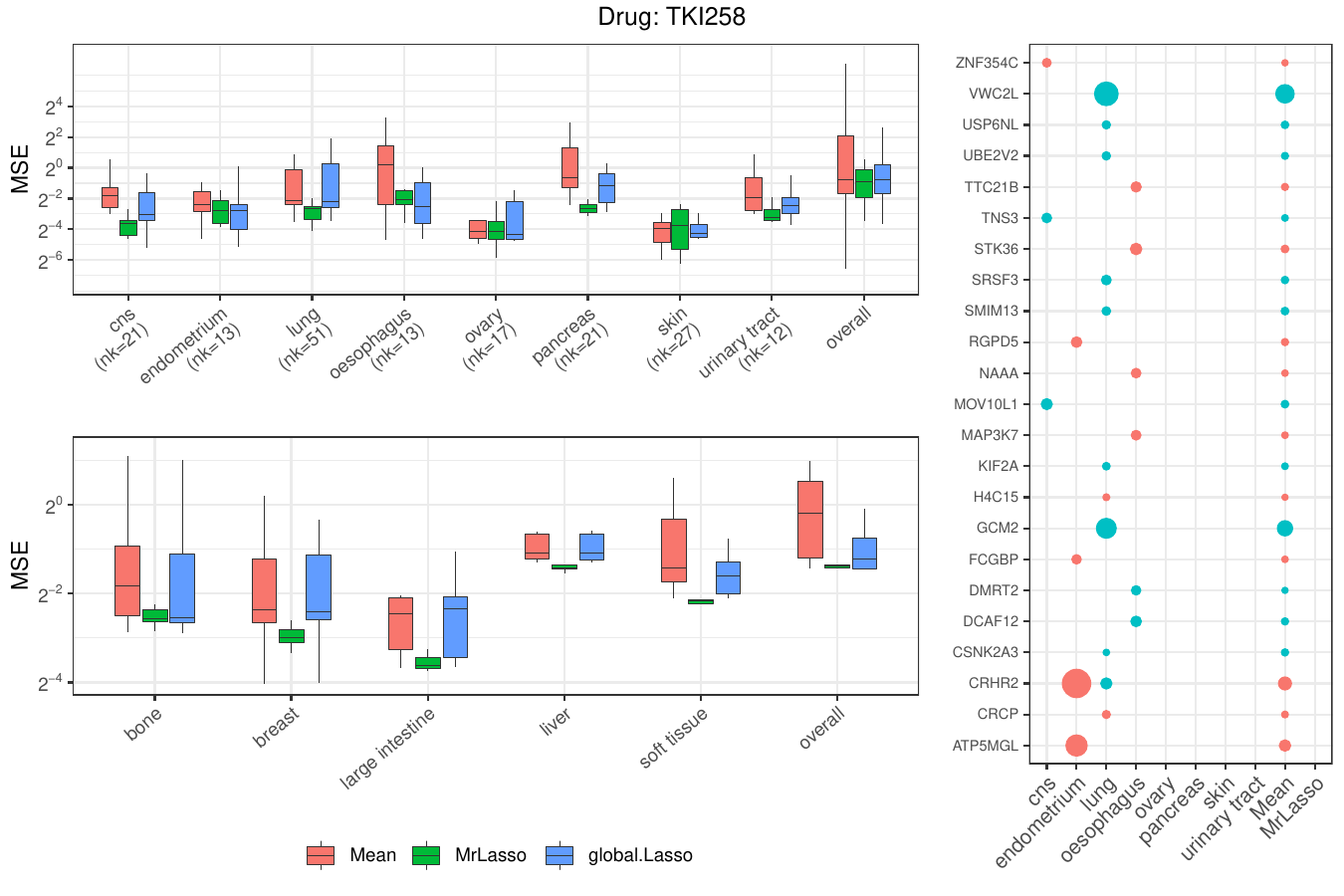}

 \caption{\textit{Left:} Drug-response prediction mean squared errors (MSEs) for  the most frequent cancer types (upper left) and rare cancer types (lower left panel) for the drug TKI258.  \textit{Right:} Coefficient plot for the selected genes in lasso estimates for most frequent cancer types and ADELE and Mr Lasso.}
 \label{fig:ccle-TKI258}
\end{figure}

\begin{figure}
 \centering
 \includegraphics[width=\textwidth]{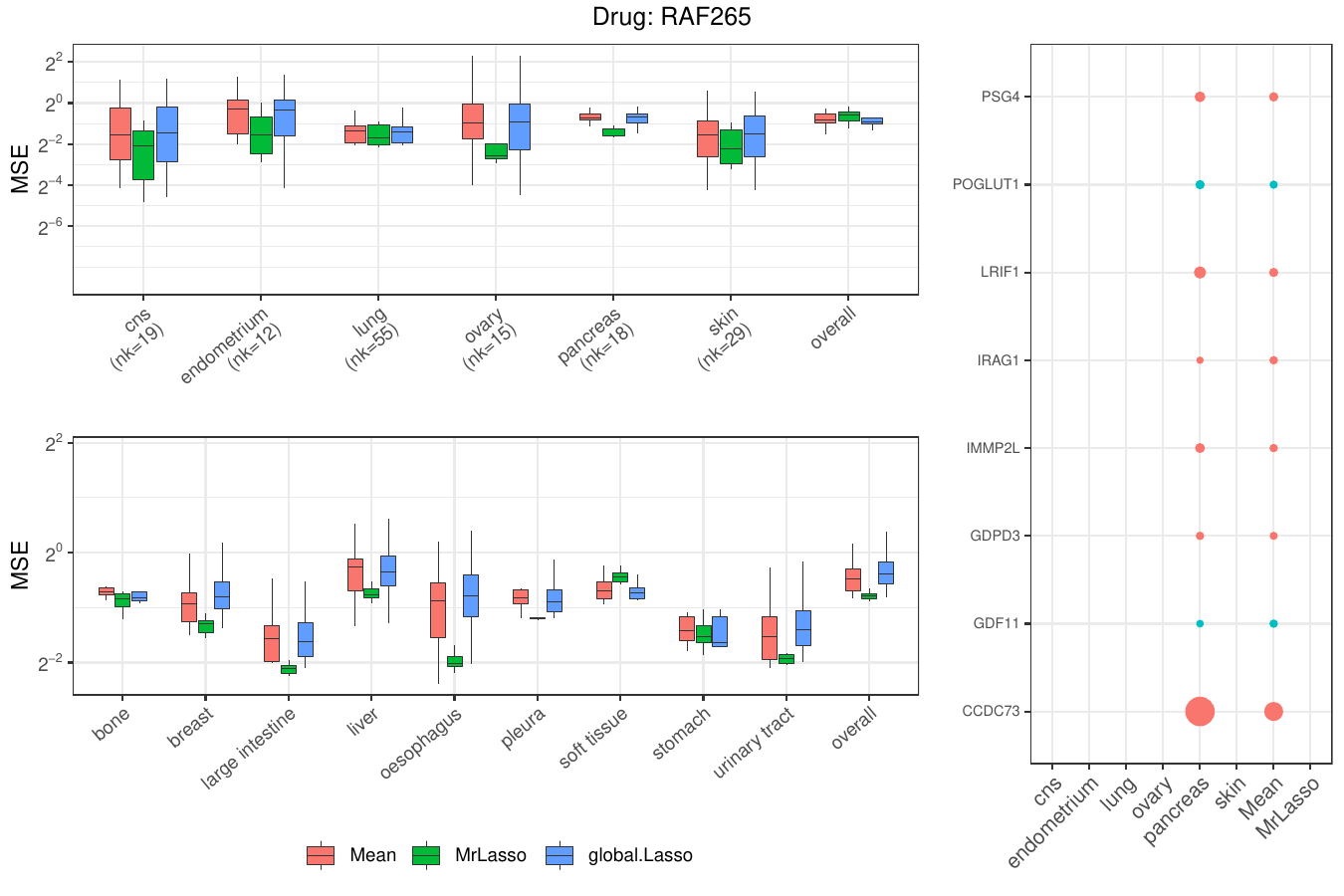}

 \caption{\textit{Left:} Drug-response prediction mean squared errors (MSEs) for  the most frequent cancer types (upper left) and rare cancer types (lower left panel) for the drug RAF265.  \textit{Right:} Coefficient plot for the selected genes in lasso estimates for most frequent cancer types and ADELE and Mr Lasso.}
 \label{fig:ccle-RAF265}
\end{figure}

\begin{figure}
 \centering
 \includegraphics[width=\textwidth]{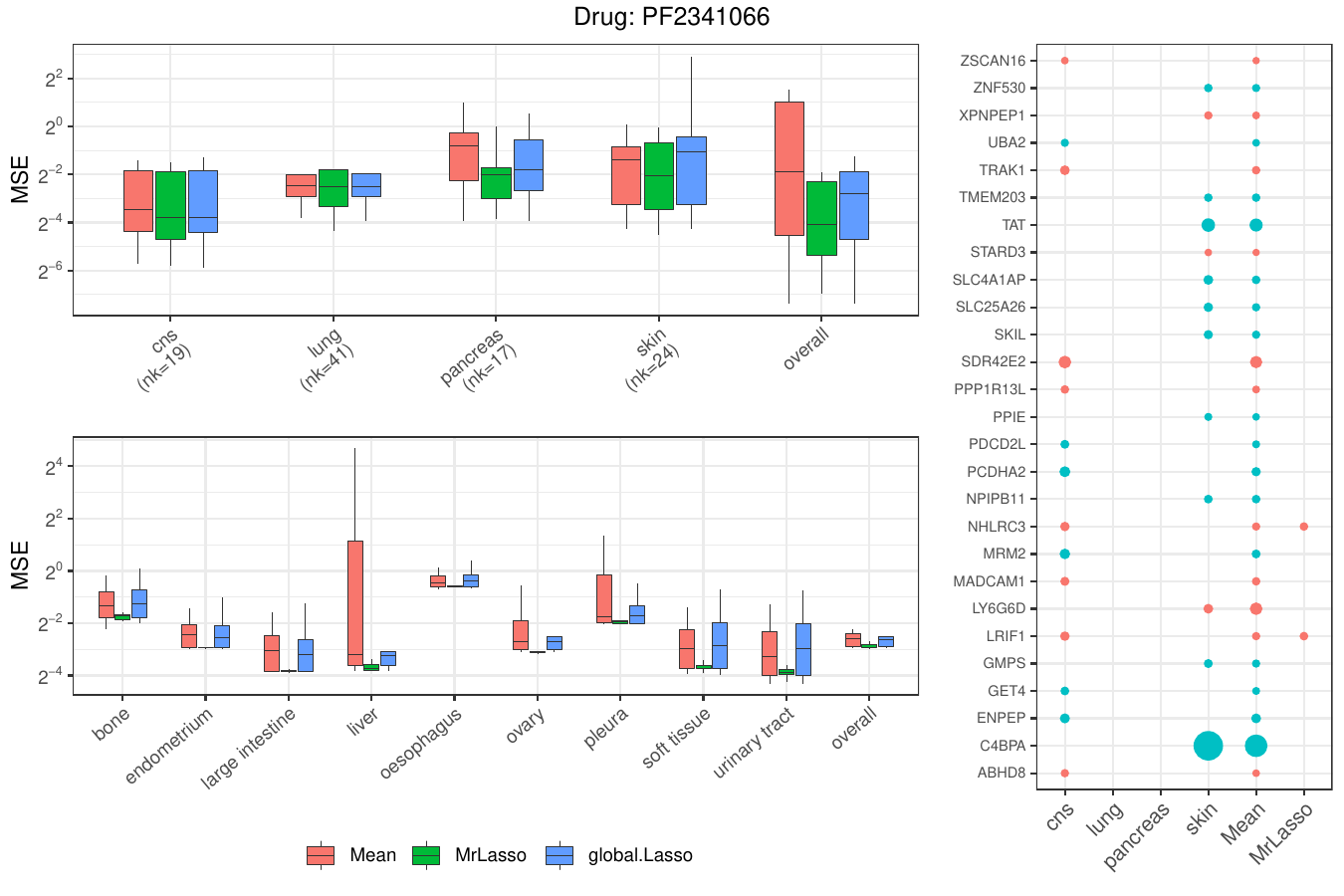}

 \caption{\textit{Left:} Drug-response prediction mean squared errors (MSEs) for  the most frequent cancer types (upper left) and rare cancer types (lower left panel) for the drug PF2341066.  \textit{Right:} Coefficient plot for the selected genes in lasso estimates for most frequent cancer types and ADELE and Mr Lasso.}
 \label{fig:ccle-PF2341066}
\end{figure}

\end{document}